\documentclass[journal,12pt,onecolumn,draftclsnofoot,]{IEEEtran}
\ifCLASSINFOpdf
\else
\fi
\usepackage[cmex10]{amsmath}
\usepackage{amsfonts}
\usepackage{subfig}
\usepackage[english]{babel}
\usepackage{amssymb,amsmath}
\usepackage{enumerate}

\usepackage{graphicx}
\usepackage{algorithm}
\usepackage{algorithmic}
\usepackage{bbm}
\usepackage{amsmath}
\usepackage{mathrsfs}
\usepackage{breqn}

\usepackage[utf8]{inputenc}
\usepackage[english]{babel}
\usepackage{amsthm}
\usepackage{enumitem}
\usepackage{booktabs}

\newtheorem{Theorem}{Theorem}
\graphicspath{ {Figures/} }

\begin{document}
\title{Joint Cell Muting  and User Scheduling in Multi-Cell Networks with Temporal Fairness}

% \author{Shahram~Shahsavari, ~\IEEEmembership{Member,~IEEE,}
%         Nail~Akar,~\IEEEmembership{Member,~IEEE,}
%         and~Babak~Hossein~Khalaj,~\IEEEmembership{Member,~IEEE}% <-this % stops a space

\author{Shahram~Shahsavari, Nail~Akar, and~Babak~Hossein~Khalaj
\thanks{S. Shahsavari is with the Department
of Electrical Engineering, New York University, New York, USA, e-mail: {\tt shahram.shahsavari@nyu.edu}}
\thanks{N. Akar is with the Electrical and Electronics Engineering Department, Bilkent University, Ankara, Turkey, e-mail: {\tt akar@ee.bilkent.edu.tr}}
% <-this % stops a space
\thanks{B. H. Khalaj is with the Electrical Engineering Department, Sharif University of Technology, Tahran, Iran, e-mail: {\tt khalaj@sharif.edu}}}

\maketitle

\begin{abstract}
A semi-centralized joint cell muting and user scheduling scheme for interference coordination in a multi-cell network is proposed under two different temporal fairness criteria. In the proposed scheme, at a decision instant, each base station (BS) in the multi-cell network employs a cell level scheduler to nominate one user for each of its inner and outer sections and their available transmission rates to a central controller which then computes the potential overall transmission rate for each muting pattern. Subsequently, the central controller selects one pattern to un-mute, out of all the available patterns. This decision is shared with all BSs which then forward data to one of the two nominated users provided the pattern they reside in, was chosen for transmission. Both user and pattern selection decisions are made on a temporal-fair basis. Although some pattern sets are easily obtainable from static frequency reuse systems, we propose a general pattern set construction algorithm in this paper. As for the first fairness criterion, all cells are assigned to receive the same temporal share with the ratio between the temporal share of a cell center section and that of the cell edge section is set to a fixed desired value for all cells. The second fairness criterion is based on the max-min temporal fairness for which the temporal share of the network-wide worst-case user is maximized. Extensive numerical results are provided to validate the effectiveness of the proposed schemes as well as to study the impact of choice of the pattern set.
%Numerical results are provided to validate the effectiveness of the proposed scheme for both criteria.
%The impact of choice of the cell muting pattern set is also studied through numerical examples for various cellular topologies.
\end{abstract}


% \fi
\IEEEpeerreviewmaketitle

\section{Introduction}
\IEEEPARstart{F}{requency} reuse of the scarce radio spectrum is key to
building high capacity wireless cellular networks \cite{hamza_etal_comsur13,kosta_etal_comsur13}. However, when frequency reuse increases as in a frequency reuse-1 LTE system for which all cells use the same band of frequencies, controlling  the adverse impact of Inter-Cell Interference (ICI) is a challenging task, especially for cell edge users. Mitigation-based Inter-Cell Interference Coordination (ICIC) consists of methods that are employed to reduce ICI through approaches such as interference cancellation and adaptive beamforming \cite{boudreau_etal_commag09}.
Avoidance-based ICIC techniques consist of frequency reuse planning algorithms through which resources are restricted or allocated to users in time and frequency domains whereas power levels are also selected with the aim of increasing SINR and network throughput
\cite{hamza_etal_comsur13,gitlin}. The focus of this paper will be on avoidance-based ICIC.

Avoidance-based ICIC schemes can be static and frequency reuse-based, or dynamic and cell coordination-based \cite{hamza_etal_comsur13}.
In static ICIC, the cell-level resource allocation is fixed and does not change over time.
The most well-known frequency reuse schemes include: (i) conventional frequency reuse schemes, (ii) fractional frequency reuse (FFR), and (iii) soft frequency reuse (SFR).

A conventional frequency reuse system with reuse factor $n>1$ statically partitions the system bandwidth into $n$ sub-bands each of which is allocated to individual cells with the reuse factor  $n$ (typical values being 3, 4, or 7) determining the distance between any two closest interfering cells using the same sub-band \cite{rappaport_book}.
% Smaller values of $L$ are indicative of denser frequency reuse.
In conventional frequency reuse systems, cell edge users are penalized due to poor channel conditions and ICI. In order to alleviate this situation, FFR partitions the system bandwidth into two groups, one for cell center (also referred to as inner or interior) users and the other for cell edge (outer or exterior) users.
The association of a given user to a cell center or cell edge section is made on the basis of its distance from the serving BS or SINR measurements
\cite{aboulhassan_etal_ntms15}. The frequency reuse parameter $n$ is set to unity for inner users whereas a  strictly larger reuse factor (typically three) is used for outer users  \cite{boudreau_etal_commag09}. It is shown in \cite{optimumFFR} that the optimal frequency reuse factor for outer users is 3. Despite increased SINR for outer users as shown in \cite{boudreau_etal_commag09},  FFR has the apparent disadvantage that a sub-band of the outer group is left unused in cell center sections. On the other hand, SFR performs per-section power allocation by assigning a lower (higher) transmission power to inner (outer) users, therefore, making it possible to use the whole frequency band for the cell center section. SFR is shown to be superior to FFR in achieving higher spectral efficiency \cite{novlan_etal_tc12,FFRSFR}.Moreover, using power optimization can improve the performance of SFR as shown in \cite{ETT:SFR}.
For other variations of static frequency reuse-based ICIC schemes, we refer to the recent survey papers on ICIC specific to OFDM-based LTE networks
\cite{hamza_etal_comsur13,kosta_etal_comsur13,aboulhassan_etal_ntms15}.

In realistic cellular wireless networks, the traffic demand is spatially inhomogeneous and changes over time. At one time, users may be concentrated in a given cluster of cells while at other times, user concentration may move to another cluster. Moreover, user distributions in the cell center and the cell edge sections may also change in time in an unpredictable way throughout the cellular network. Therefore, static frequency reuse-based ICIC schemes fall short in coping with dynamic
workloads in time and space and methods are proposed for dynamic workloads in \cite{gitlin,gonzales_etal_2010,li_etal_tett14}. Dynamic ICIC (D-ICIC), on the other hand, relies on cell coordination-based methods that dynamically react to changes in traffic demands and user distributions \cite{ETTinterference}. Despite the apparent theoretical advantage of D-ICIC network throughput, most proposed schemes in this category are very complex to implement. D-ICIC schemes are categorized into centralized, semi-centralized, and de-centralized, on the basis of how cell-level coordination is achieved and subsequently the complexity of implementation of the underlying scheme \cite{gonzales_etal_vtc10}.
In centralized D-ICIC, the channel state information of each user is fed to a centralized entity which then makes scheduling decisions  to maximize the throughput under fairness and power constraints \cite{bohge_etal_ewc09}.  However, such centralized scheduling is complex to implement due to the requirement of timely and large per-user feedback information as well as the complexity of the centralized scheduler \cite{corvino_etal_wcnc09}. Semi-centralized schemes use centralized entities which only perform cell-level coordination while user-level allocation is performed by each BS \cite{li_liu_vtc03,rahman_yanikomeroglu_twc10}.
The reference \cite{li_liu_twc06} considers a semi-centralized radio resource allocation scheme in
OFDMA networks where radio resource allocation is performed at two layers. At the upper layer,
a centralized algorithm
coordinates ICI between BSs at the super-frame level and each
BS makes its scheduling decisions opportunistically based on instantaneous channel conditions of users.
The computational complexity of semi-centralized schemes are much less than centralized schemes with reduced feedback requirement; however, a sufficiently low-delay infrastructure is still needed.  In de-centralized D-ICIC, there is no centralized entity but a local signaling exchange is still needed among BSs \cite{gonzales_etal_vtc10}. The focus of this paper is the multi-cell scheduling problem with semi-centralized D-ICIC which can be implemented provided a low-delay and efficient backhaul exists.

In non-dense frequency reuse-$n$ networks with $n>1$, a single-cell scheduler
decides which user to schedule without a need for coordination among cells.
In this paper, we focus on opportunistic scheduling with fairness constraints. In opportunistic scheduling, the scheduler tries
to select a user having the best channel condition
at a given time \cite{kulkarni_rosenberg_wirnet05}. Such a greedy opportunistic scheduler would maximize the throughput however fairness among users would not be achieved.
Practical opportunistic schedulers exploit the time-varying nature of the wireless channels for maximizing cell throughput under certain fairness constraints.
In Proportional Fair (PF) scheduling, at a scheduling instant, the BS chooses to serve the user which has the largest ratio of available transmission rate to its exponentially smoothed average throughput \cite{jalali,tsePF}.
Different variations of the PF algorithm are possible depending on how the scheduler treats empty or short queues and how the average throughput is maintained \cite{andrews}. In Temporal Fair (TF) single-cell opportunistic scheduling, the cell throughput is maximized under the constraint that users receive the same temporal share, i.e., the same average air-time \cite{shroff_jsac01}. The work in
\cite{shroff_jsac01} shows that the optimum single-cell TF scheduler chooses the user which has the largest sum of available transmission rate and another user-dependent term when an appropriate channel model is available, or alternatively, this additional term can be obtained using an on-line learning algorithm.
Under some simplifying assumptions involving channel characteristics of users, the PF and TF methods are shown to be equivalent \cite{tsePF,holtzman_waterfilling}.
We refer the reader to \cite{survey_downlink} for a survey on single-cell scheduling in LTE networks.  There have been few studies to generalize
single-cell fairness to multi-cell or network-wide fairness in multi-cell networks.
Reference \cite{cho_etal_twc09} shows via simulations that network-wide opportunistic scheduling and power control is effective
for fairness-oriented networks. The authors in \cite{shahsavari_akar_wcl15} propose a semi-centralized approach to achieve inter-cell and intra-cell temporal fairness in multi-cell networks but user-level network-wide fairness is not studied in that work.

In this paper, we propose a semi-centralized joint cell muting and user scheduling scheme for interference coordination in the downlink of a multi-cell network under two different temporal fairness criteria. In this scheme, the network operator is given a set of cell muting patterns each of which is associated with a set of cells that can transmit simultaneously with acceptable ICI in the multi-cell network. The scheduler operates in two levels (upper and lower levels) as in  \cite{li_liu_twc06} as follows.
At a scheduling instant, each BS employs a cell level (lower level) TF scheduler to nominate one user for each of its inner and outer sections and their available transmission rates to a central entity which then computes the potential overall transmission rate for each muting pattern. Subsequently, a TF scheduler is run at the central entity (upper level) to decide which pattern to activate. This decision is then shared with all BSs which then forward data to one of the two nominated users provided the pattern they reside in, was chosen for transmission.

The upper and lower level TF schedulers can be tuned to conform to two different temporal fairness criteria. As for the first fairness criterion, all cells receive the same temporal share with the ratio between the temporal share of a cell center section and that of the cell edge section
is set to a fixed desired value for all cells. Within a section, all users receive the same temporal share. Although the first fairness criterion achieves fairness among cells and also users within the same section, this criterion does not seek network-wide user fairness. As a remedy, we propose to use network-wide max-min temporal fairness as the second fairness criterion.
A broad range of centralized and/or distributed algorithms are available in the literature to implement max-min fairness in the context of sharing resources including link bandwidth \cite{parekh_gallager_ton93}, network bandwidth \cite{bertsekas_gallager_book}, CPU \cite{stoica_etal_rtss96}, and cloud computing \cite{ghodsi}. The high popularity of the notion of max-min fairness in general computing and communication systems has led us to study in this paper the network-wide max-min temporal fairness for which the temporal share of the network-wide worst-case user in the multi-cell network is maximized. To the best of our knowledge, network-wide max-min fairness in this context has not been studied before.
Both fairness criteria are shown to be handled within the framework we propose in this paper.
As a further contribution, we propose a novel general pattern set construction algorithm  with reasonable computational complexity
using fractional frequency reuse principles with $M$ cell muting patterns.
The complexity of the proposed upper level TF scheduler turns out to be the same as that of a single-cell TF scheduler with $M$ users and is therefore quite efficient for relatively small cardinality parameter $M$.
The impact of choice of the cell muting pattern set and its cardinality $M$ is also studied through numerical examples for various cellular topologies.

The proposed approach leads to reduced computational complexity of the upper level scheduler and reduced information exchange requirements between BSs and the central entity in comparison with centralized schemes that have higher implementation complexities \cite{bohge_etal_ewc09}. Although most of the literature on the interference coordination techniques are based on OFDMA-based  LTE networks, we consider a time-slotted single-carrier air interface for the sake of simplicity.
However, it is possible to extend the proposed framework to OFDMA-based networks.
As an example, the case where all resource blocks (RB) in an LTE frame are assigned by the BS to one single user only, can be obtained through an immediate extension.
However, further generalizations including service to multiple users within a single frame are left out of the scope of this study.

The paper is organized as follows. We present the proposed multi-cell architecture along with the descriptions of cell muting patterns and pattern set construction algorithms in Section \ref{section2}.
The two forms of fairness criteria that we employ as well as the two-level scheduler proposed to satisfy both criteria are presented in Section \ref{section3}. We validate the effectiveness of the proposed
approach in Section \ref{section4}. Finally, we conclude.

\section{Proposed Multi-Cellular Architecture}
\label{section2}
\subsection{Cells and Users}
We consider the downlink of a time-slotted single-carrier frequency reuse-1 un-sectored Cellular Network (CN) with bandwidth $BW$ where the time slots are indexed by $1\leq \tau <\infty$.
We assume that each cell is divided into inner  and outer sections;  see Fig.~\ref{cell-sections} for a cell with cell radius $R_C$ and inner section radius $R_I$. Let $C_i$, $i=1,2,\ldots,K$, denote cell $i$ in the CN where $K$ is the total number of cells.  Let $BS_i$  denote the base station located in $C_i$. Also let $I_i$ and $O_i$ denote the inner and outer sections of $C_i$, respectively.  We let $N$ denote the total number of users in the network and let $N_i$, $N^I_i$ and $N^O_i$ denote the number of users associated with cell $C_i$, and with sections $I_i$ and $O_i$, respectively.  The cell or section association of a given user is assumed to be a-priori known throughout this paper.
Obviously, $N_i=N^I_i+N^O_i$ and $N = \sum_i N_i$. Let $U^I_{i,j}$, $j = 1,2,\ldots,N^I_i$ and $U^O_{i,j}$, $j = 1,2,\ldots,N^O_i$ denote the user $j$ associated
with $I_i$ and $O_i$, respectively. We assume all users are persistent, i.e., they always have data to receive.  For a given time slot $\tau$, the cell $C_i$ is active (un-muted) if its $BS_i$ is transmitting. Otherwise, $C_i$ is said to be muted. When $C_i$ is active and $BS_i$ is transmitting to a user in $I_i$, then $I_i$ is called active; otherwise  $O_i$ is active.
\subsection{Cell Muting Patterns}
A  cell muting or transmission pattern (or pattern in short) is defined as a subset of the set of all sections in the CN satisfying the following two properties:
\begin{itemize}%[leftmargin=*]
\item \noindent Patterns are noise-limited as opposed to being interference-limited, i.e., the elements of a pattern can be activated simultaneously without the associated base stations creating destructive interference on users associated with other active cells. Obviously, when a pattern is selected by the scheduler, all cells which do not have any inner or outer sections in that chosen pattern, would be muted.
\item  Patterns are maximal, i.e., a pattern can not be included in another pattern with larger cardinality.
\end{itemize}
Obviously, a pattern is governed by the geometry of the CN, the power levels used by a BS for transmitting to inner and outer section users,  and the definition of destructive interference.

In this paper, patterns are constructed on the basis of an underlying  Fractional Frequency Reuse-$n$ (FFR-$n$)-type CN; see \cite{bilios,novlan_etal_globecom10}.
As an example, Fig.~\ref{patternFFR3} illustrates four patterns in a 9-cell network inspired by an FFR-3 system \cite{bilios}.
In conventional FFR-3 systems, a separate frequency band is assigned to each of the four patterns which are activated simultaneously in time but in different frequency bands. In our proposed architecture, the entire bandwidth $BW$ of the frequency reuse-1 system is dynamically shared by the available patterns in time and not in frequency, by means of dynamically muting all but one pattern at a given time slot.  It is clear that deployment of a more general FFR-$n$ system with $n=x^2 + xy + y^2$ for some non-negative integers $x$ and $y$, can provide $M=n+1$ patterns for $x,y\geq 1$ \cite{rappaport_book}. We call this set of $n+1$ patterns an Essential Pattern Set (EPS-$n$) for the associated FFR-$n$ system.
On the other hand, a General Pattern Set (GPS) is an arbitrary collection of patterns in which each section in the CN is an element of at least one pattern.
It is clear that EPS-$n$ is a GPS. The Universal Pattern Set (UPS) is the set of all possible patterns.
Obviously, GPSs are subsets of the UPS.
% A mechanism to generate a GPS systematically will be discussed in the next section and hereafter in this section, we assume that the GPS  $\{P_m\}_{m=1}^M$ is a-priori
% known.
A pattern is said to be active at a given slot if all the sections included in the pattern are active. A section may be included in multiple patterns for a GPS. However, for EPS-$n$, the patterns are mutually exclusive.  Fig.~\ref{GPS-eg} illustrates a sample GPS for the 9-cell CN; note that $I_9$ is an element of the three patterns $P_1$, $P_3$, and $P_4$ for this sample GPS.
\begin{figure}[th]
\centering
\includegraphics[width=4cm]{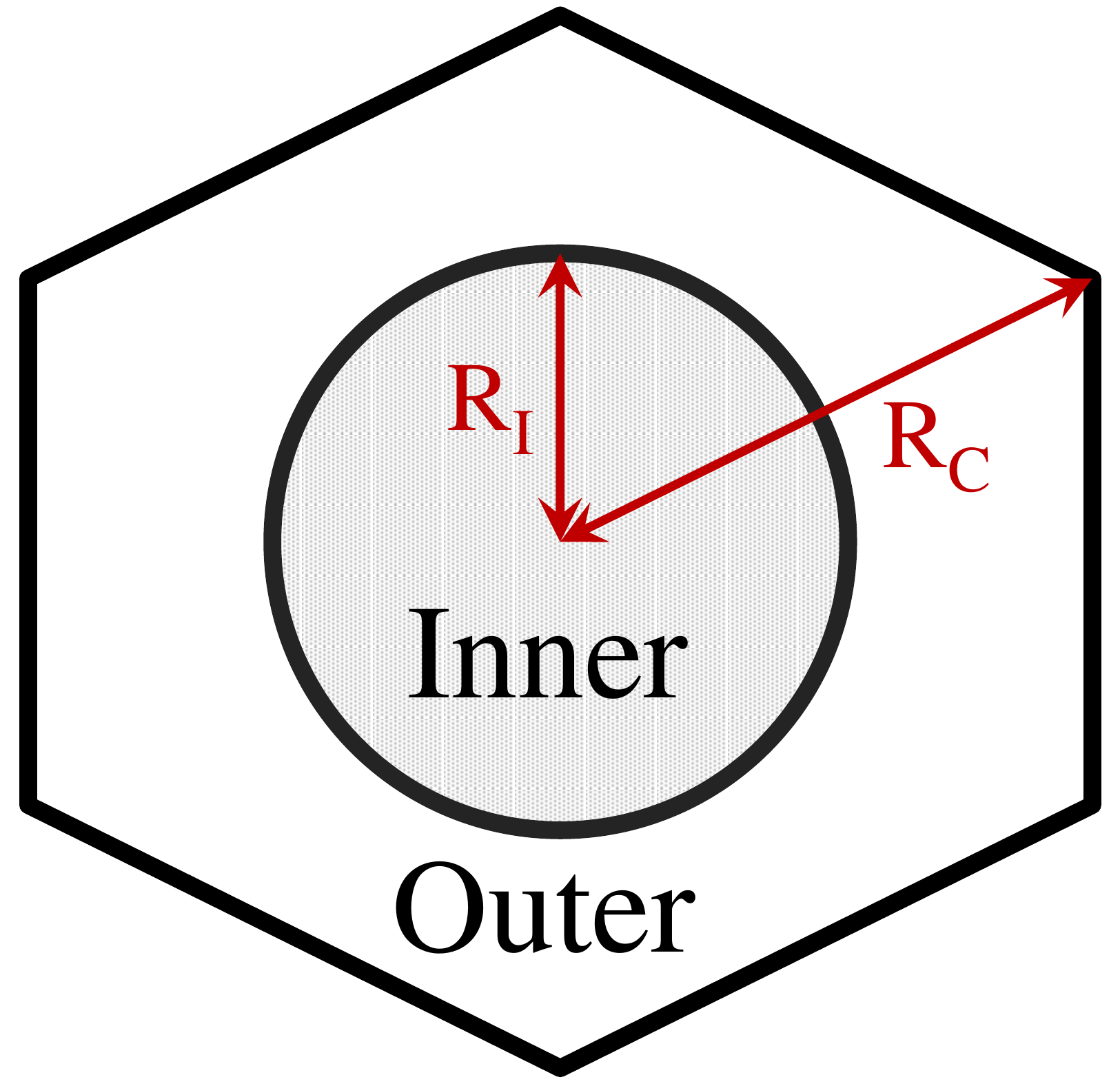} % 2.5 cm
\caption{Inner and outer sections of a cell}
\label{cell-sections}
\end{figure}

\begin{figure}[th]
\centering
\includegraphics[width=0.7\linewidth]{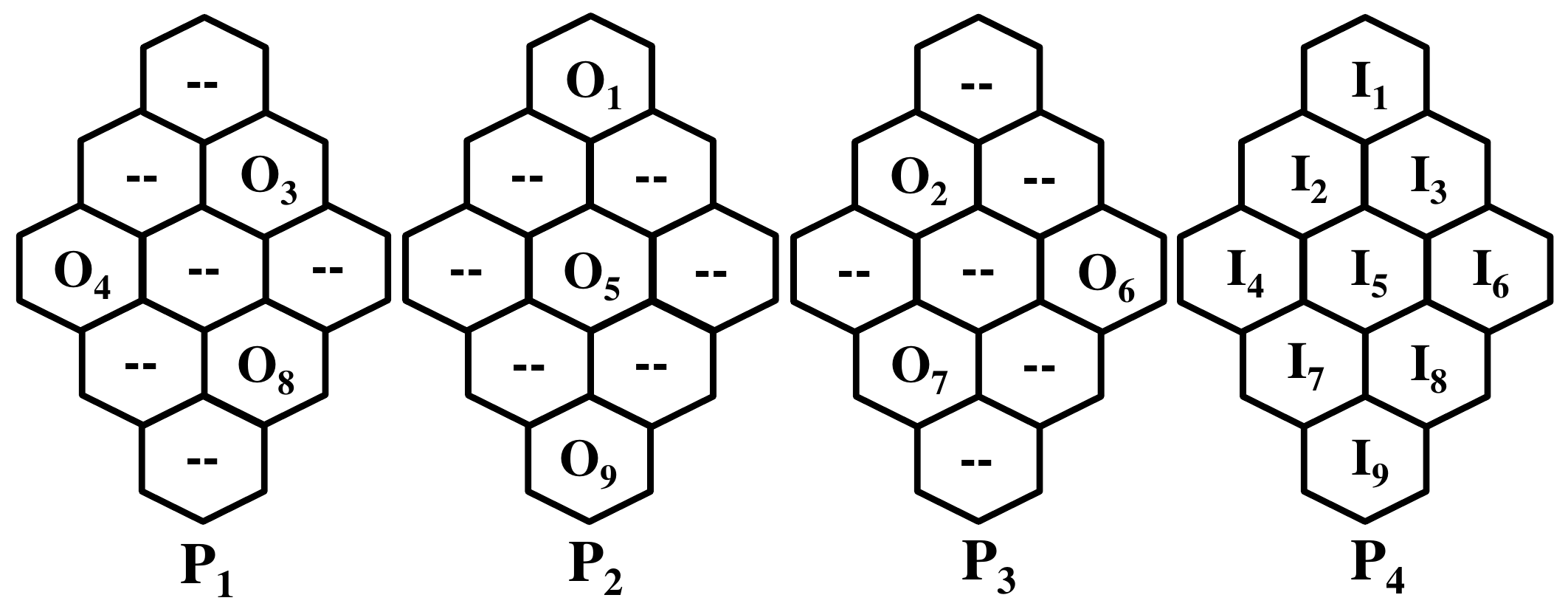}
\caption{EPS-3 for  the 9-cell CN:  $ P_1=\{O_3,O_4,O_8\}, P_2=\{O_1,O_5,O_9\}, P_3=\{O_2,O_6,O_7\}, P_4=\{I_1, I_2, \ldots, I_9\}.$}
\label{patternFFR3}
\end{figure}

\begin{figure}[th]
\centering
\includegraphics[width=0.55\linewidth]{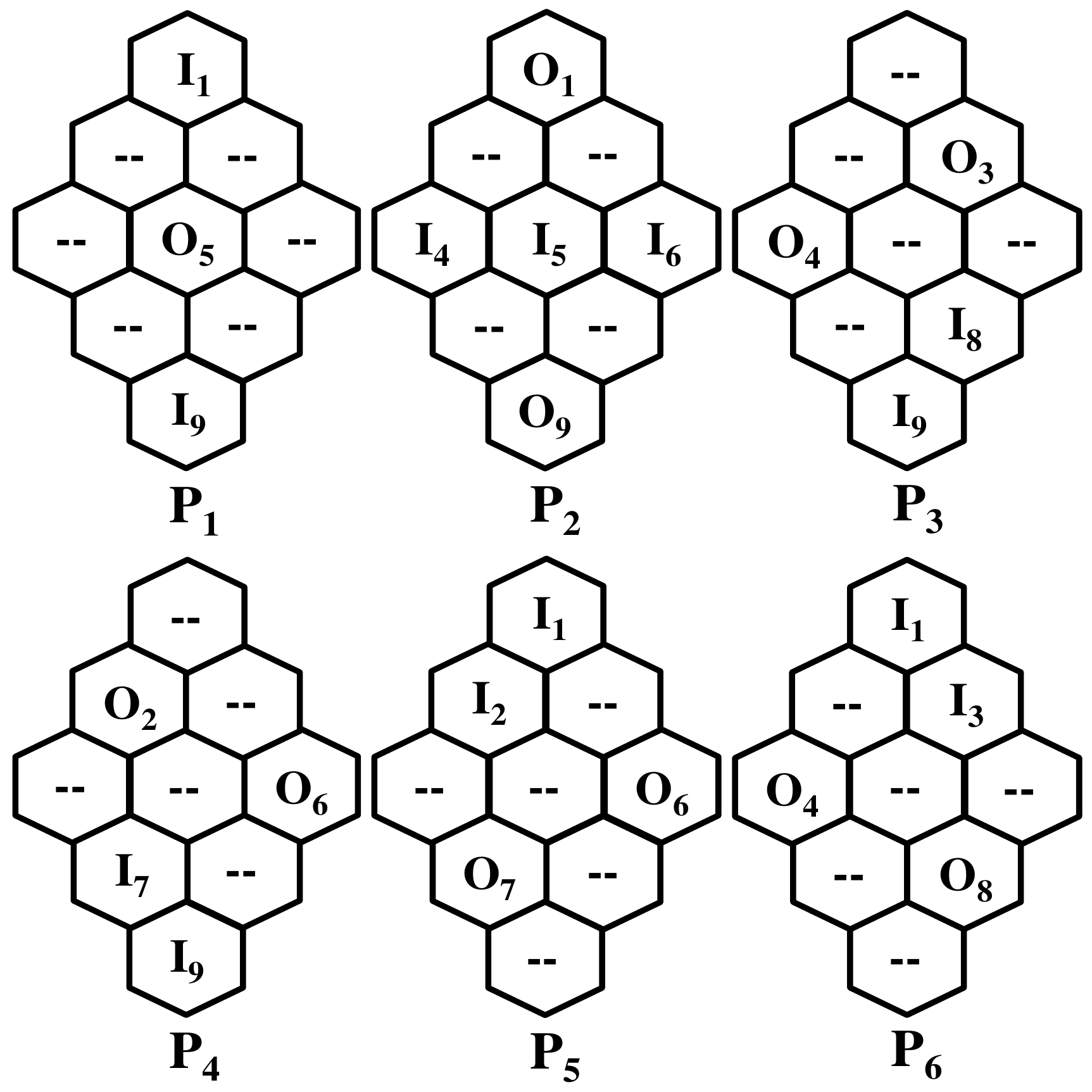}
\caption{A sample GPS with 6 patterns for the 9-cell CN: $P_1=\{I_1,I_9,O_5\}$, $P_2=\{I_4,I_5,I_6,O_1,O_9\}$, $P_3=\{I_8,I_9,O_3,O_4\}$, $P_4=\{I_7,I_9,O_2,O_6\}$, $P_5=\{I_1,I_2,O_6,O_7\}$, $P_6=\{I_1,I_3,O_4,O_8\}$.}
\label{GPS-eg}
\end{figure}
\subsection{General Pattern Set Construction Mechanism}
\label{subsec:pattern set construction}
In this section, we propose an algorithmic method to construct a GPS which can effectively be used in multi-cell networks.
For this purpose, let the operator $\overset{n}{\leftrightarrow}$ represent the interference relationship in an FFR-$n$ system with $X \overset{n}{\leftrightarrow} Y$ implying the activation of sections $X$ and $Y$ causing destructive interference of $X$ on $Y$ and vice versa.  In the lack of destructive interference between the two associated cells, we say $X \overset{n}{\not\leftrightarrow} Y$. We assume that the interference is not destructive if its power is much less than the noise power.
The interference relationship which holds for any two given sections $X$ and $Y$ depends on the following two parameters:
\begin{enumerate}[label=\roman*),leftmargin=*]
\item The physical distance between the BSs to whom sections $X$ and $Y$ are associated with,
\item The transmission power levels assigned to inner and outer section users.
\end{enumerate}
In order to quantify interference between two sections, we make use of the essential patterns in an FFR-$n$ system. Among the patterns in ESP-$n$, except one transmission pattern which consists of all inner sections in the network (for example $P_4$ in Fig. \ref{patternFFR3}), there are $n$ patterns each of which consists of a number of outer sections in the network. We call all such $n$ patterns as a Mother Pattern Set (MPS-$n$). For instance, MPS-3 denotes the set
$\{P_1,P_2,P_3\}$; see Fig.~\ref{patternFFR3}.
% We will later use MPS-$n$ to construct a specific GPS that can effectively be used in multi-cell networks.
Let $D_n$ denote the distance between the centers of the nearest active cells in the patterns of MPS-$n$. Based on \cite{rappaport_book}, we have $D_n=R_C\sqrt{3n}$. Let us denote by $d_{i,j}$ the distance between the centers of $C_i$ and $C_j$. Consequently, we have the following interference identities:
\begin{eqnarray}
O_i  \overset{n}{\leftrightarrow} O_j~~& \text{iff} & ~~ d_{i,j}<D_n \label{O-Oneighborhood},\\
I_i \overset{n}{\leftrightarrow} O_j~~& \text{iff} & ~~ d_{i,j}<D_n \label{I-Oneighborhood}, \\
I_i \overset{n}{\not\leftrightarrow} I_j & ~\text{for all}&   i,j \label{I-Ineighborhood}.
\end{eqnarray}
% In practice, BS transmit power to the inner users is designed in a way that inter-cell interference on inner users in other cells is negligible compared to the noise power. Therefore, as (\ref{I-Ineighborhood}) suggests, any inner section does not have any neighbor among other inner sections in the network. However, because BS transmit power to outer users is higher and those are more sensitive to interference, each outer section has some neighbors that can cause interference for them. Using (\ref{O-Oneighborhood})-(\ref{I-Ineighborhood}),
Moreover, we consider different fixed downlink power levels for inner and outer section users, designed in a way that (\ref{O-Oneighborhood})-(\ref{I-Ineighborhood}) hold. We note that $\forall i:~I_i \overset{n}{\leftrightarrow} O_i$ based on (\ref{I-Oneighborhood}) meaning that at most one section per cell can be active at a time slot. The following theorem (given without a proof since it is relatively straightforward) provides the necessary and sufficient conditions for a set of sections to constitute a pattern in an associated FFR-$n$ CN with $K$ cells.
\begin{Theorem}
A set of sections $P\neq\varnothing$ amounts to a pattern if and only if the following conditions hold for any $i,j \in\{1,2,\ldots, K\}$:
\begin{enumerate}[label=\alph*),leftmargin=*]
\item $O_i \overset{n}{\leftrightarrow} O_j, O_i \in P \Rightarrow O_j \not\in P$,
\item $O_i \overset{n}{\leftrightarrow} I_j, O_i \in P \Rightarrow I_j \not\in P$,
\item $I_i \overset{n}{\leftrightarrow} O_j, I_i \in P \Rightarrow O_j \not\in P$,
\item $\{I_i \in P \} \vee \{\exists k:O_k \overset{n}{\leftrightarrow} I_i, O_k \in P \}$,
\item $O_i \in P, O_i \overset{n}{\not\leftrightarrow} I_j \Rightarrow \{I_j \in P\} \vee \{\exists k: O_k \overset{n}{\leftrightarrow} I_j, O_k \in P\} $
\item  $I_i \in P, I_i \overset{n}{\not\leftrightarrow} O_j \Rightarrow \{O_j \in P\} \vee \{\exists k: O_k \overset{n}{\leftrightarrow} O_j, O_k \in P \} \vee \{\exists k: I_k \overset{n}{\leftrightarrow} O_j, I_k \in P \}$
\item $O_i \in P, O_i \overset{n}{\not\leftrightarrow} O_j \Rightarrow \{O_j \in P\} \vee \{\exists k: O_k \overset{n}{\leftrightarrow} O_j, O_k \in P\} \vee \{\exists k: I_k \overset{n}{\leftrightarrow} O_j, I_k \in P \}$
\end{enumerate}
\label{suff-nece}
\end{Theorem}

In Theorem \ref{suff-nece}, the first three conditions are required for a pattern to be noise-limited
whereas the remaining conditions are required for a pattern to be maximal. One straightforward way of constructing all transmission patterns (UPS) is an Exhaustive Search (ES) among all possible sets of sections of cardinality $3^K$ since each cell's inner section or outer section is included in a given set or not, leading to three possibilities for each cell and there are $K$ such cells.  In the ES method, all of these sets are generated first, then using Theorem \ref{suff-nece}, one can check whether each set of sections is a pattern or not in a loop. Consequently, in order to check the validity of all the generated sets of sections, the ES method requires $\mathcal{O}(3^K)$ iterations which is not feasible to run for large networks.

In this paper, we propose a novel algorithm with reduced computational complexity to construct a GPS  (using MPS-$n$) which is a subset of the UPS.
 For this purpose, we let $\{P^*_m\}_{m=1}^n$ be the patterns of MPS-$n$. Also, let $S^m_k, k=1,2,\ldots, 2^{\left\vert{P^*_m}\right\vert}$ denote the subset $k$ of the set of sections included in $P^*_m$ where  ${\left\vert{P^*_m}\right\vert}$ denotes the number of sections included in $P^*_m$.
For convenience, let $S^m_1=\emptyset$ for each $m$. Algorithm~\ref{pattern-generation-alg} proposes an algorithmic method to generate a GPS using MPS-$n$ for an arbitrary value of $n$.
 Fig.~\ref{pattern-generation-eg} illustrates the working principle of Algorithm~\ref{pattern-generation-alg} for $n=3$.
In Fig.~\ref{pattern-generation-eg}a, one of the mother patterns in the FFR-3 system is depicted.
 Fig.~\ref{pattern-generation-eg}b depicts one of the subsets of that mother pattern. Due to interference, neither the inner nor the outer sections of the shadowed cells can be a member of the new pattern $V$  in Fig.~\ref{pattern-generation-eg}c based on (\ref{O-Oneighborhood}) and (\ref{I-Oneighborhood}). In Fig.~\ref{pattern-generation-eg}d, inner sections of the remaining cells are added to the pattern $V$ to attain maximality based on (\ref{I-Ineighborhood}). It is clear that EPS-$n$ is a subset of the pattern set constructed by this algorithm. Next, we present
Theorem~\ref{alg1-gps} along with its proof showing that the pattern set produced by Algorithm ~\ref{pattern-generation-alg} is indeed a GPS.
% This figure shows that how a different transmission pattern can be formed using one of patterns in MPS-3. We should note that EPS-n $\in$ PS generated patterns by Algorithm \ref{pattern-generation-alg}. The reason is that empty set is a subset of $P^*_m$ for every $m$ and it is mapped to the pattern $V=\{I_i\}_{i=1}^K$ in Algorithm \ref{pattern-generation-alg}. Moreover, as a set, every $P^*_m$ is a subset of itself and Algorithm \ref{pattern-generation-alg} maps it to itself. Therefore, the output PS of Algorithm \ref{pattern-generation-alg} includes EPS-n. In the next theorem, we prove that output PS of Algorithm \ref{pattern-generation-alg} is a GPS.

\begin{algorithm}
\caption{Constructing a GPS using MPS-n}
 \label{pattern-generation-alg}
\textbf{Input}: MPS-$n$\\
\textbf{Output}: GPS
\begin{algorithmic}[1]
\STATE  GPS $\leftarrow \varnothing$\;
\FOR {$m=1$ to $n$ step $1$}
        \FOR {$k=1$ to $2^{\left\vert{P_m^*}\right\vert}$ }
            \STATE $V \leftarrow S^m_k$\;
            \STATE Add $I_i,~i=1,2,\ldots,K$ to $V$ unless \\$\exists j:\{O_j \in V\} \wedge \{I_i \overset{n}{\leftrightarrow} O_j$\}\;
            \IF {$V \not\in$ GPS}
                \STATE Add $V$ to the GPS;
            \ENDIF
        \ENDFOR
\ENDFOR
\end{algorithmic}
\end{algorithm}

\begin{figure}[h]
\centering
\includegraphics[width=0.75\linewidth]{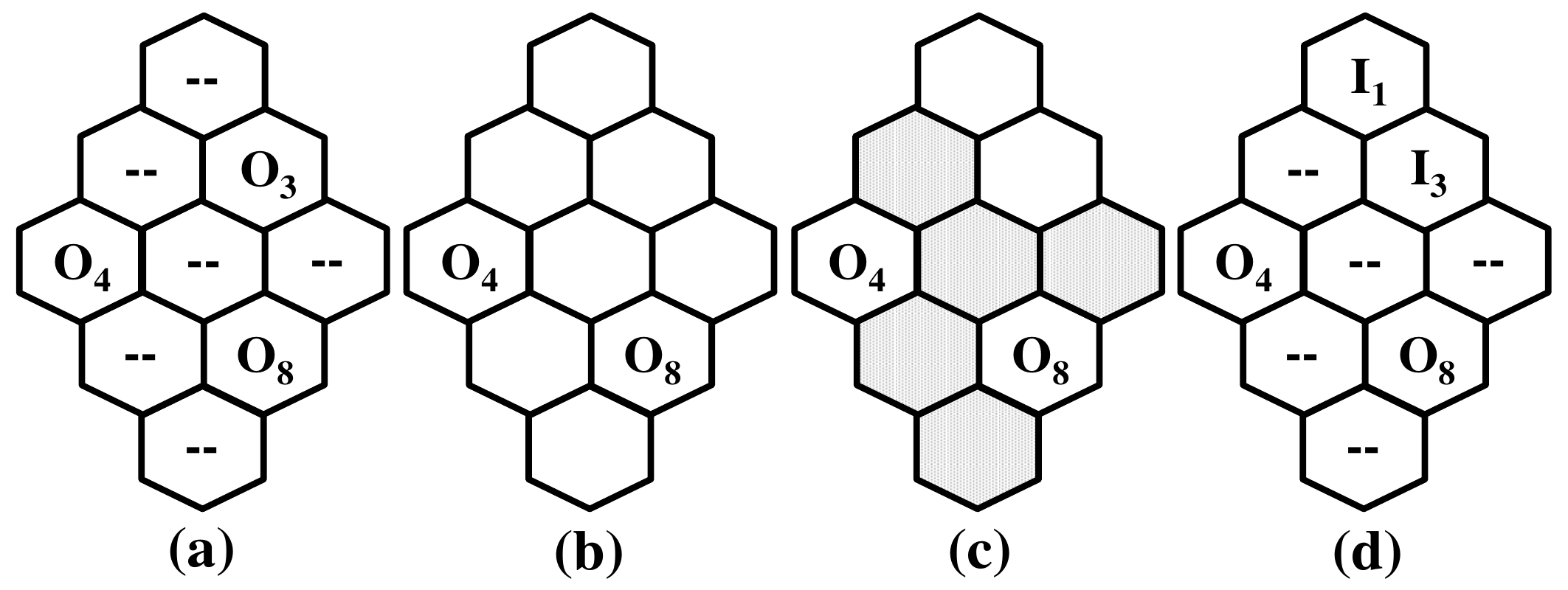}
\caption{Steps to construct a transmission pattern from a mother pattern in an FFR-3 system.
% (a) One of the mother patterns in FFR-3 system  (b) One of the subsets of that mother pattern (c) Because of the severe interference, neither the inner sections nor the outer sections of the shadowed cells can be a member of the new pattern V  (d) Inner sections of the remain cells are added to the new pattern V, to preserve the maximality
}
\label{pattern-generation-eg}
\end{figure}

\begin{Theorem}
The pattern set produced by Algorithm~\ref{pattern-generation-alg} is a GPS.
\label{alg1-gps}
\end{Theorem}

\renewcommand{\qedsymbol}{\rule{0.5em}{0.5em}}
\begin{proof}
We need to show that the elements of the pattern set produced by Algorithm~\ref{pattern-generation-alg} are noise-limited and maximal and moreover each section in the CN is an element of at least one pattern in the pattern set. It is clear that all the patterns are noise-limited due to the way a new pattern, $V$, is generated by Algorithm~\ref{pattern-generation-alg}.
On the other hand, assume at least one of the generated patterns is not maximal meaning that we can at least add one other section (say section $X$) to it, while it remains noise-limited. Clearly, section $X$ cannot be an inner section, because Algorithm~\ref{pattern-generation-alg} adds all possible inner sections. On the other hand, $X$ can not be an outer section because assume that $X$ is the outer section of cell $Y$ with the inner section $Z$. If adding $X$ would not cause destructive interference, then adding $Z$ also would not cause destructive interference. However, this contradicts with the way the algorithm works. To show that each section is an element of at least one pattern, recall that patterns which belong to MPS-$n$ are among the generated patterns and each of the outer sections is an element of an individual MPS. Moreover, the set of all inner sections is obtained  as a pattern when we start from the subset $S^m_1$ for any $m$. This concludes the proof.
\end{proof}

We now elaborate on the computational aspects of the GPS construction algorithm that we have proposed. The outer loop of Algorithm~\ref{pattern-generation-alg} requires
$n$ iterations and the inner loop requires approximately $2^{\lceil{K/n}\rceil}$ iterations. To see this, there are $n$ mother patterns, namely $\{P^*_m\}_{m=1}^n$, and each mother pattern $P_m^*$ consists of approximately $\left\vert{P_m^*}\right\vert\approx\lceil{K/n}\rceil$ outer sections, so each inner loop will be executed
approximately $2^{\lceil{K/n}\rceil}$ times. Therefore, the total number of required iterations is $\mathcal{O} (2^{\lceil{K/n}\rceil}n )$ which is substantially less than that of the ES method.
% Furthermore, while changing parameter $n$ does not have any effect on the number of iterations in exhaustive search (because it is not involved there), increasing $n$ decreases the number of iterations in Algorithm \ref{pattern-generation-alg} and as a result, it decreases the time complexity of this algorithm. However, by increasing $n$, the number of generated patterns also decreases, because roughly speaking, one pattern is generated per iteration in Algorithm \ref{pattern-generation-alg}. We should also note that the amount of per iteration computations of Algorithm \ref{pattern-generation-alg} is much lower than this value for ES, because in each iteration of  Algorithm \ref{pattern-generation-alg}, some inner sections are added to a pattern, but in exhaustive search method, all of the conditions of Theorem \ref{suff-nece} should be checked for every $i$ and $j$.
Moreover, the per-iteration computational load of Algorithm~\ref{pattern-generation-alg} is less than that of ES. However, recall that this GPS may not be the same as UPS and there may be patterns in the UPS that can not be constructed by the proposed algorithm.
To evaluate the pattern construction capability of Algorithm~\ref{pattern-generation-alg}, we consider two different CNs illustrated
in Fig.~\ref{three-scenarios} depicting (a) 9-cell and (b) 6-cell scenarios, which will be used in numerical examples.
We use the FFR-3 interference relationships and subsequently MPS-3 as input to the proposed algorithm.  The
ES and Algorithm~\ref{pattern-generation-alg} are run for each of the two scenarios and
the patterns constructed by each method are presented in
Tables~\ref{patterns-42-22} and \ref{patterns-13-10}, respectively, for the 9-cell and 6-cell scenarios.
We observe that ES generates the UPS with cardinality $42$ and $13$, whereas Algorithm~\ref{pattern-generation-alg} constructs
$22$ and $10$ patterns, respectively, for the 9-cell and 6-cell scenarios.
%Note that the GPS and UPS are identical for the 3-cell scenario.
Although the proposed algorithm can not construct the UPS, we will later show through numerical examples that the network performance obtained by the GPS produced by the proposed algorithm is only slightly inferior to that attained by that of the UPS.

\begin{figure}[t]
\centering
\includegraphics[width=0.5\linewidth]{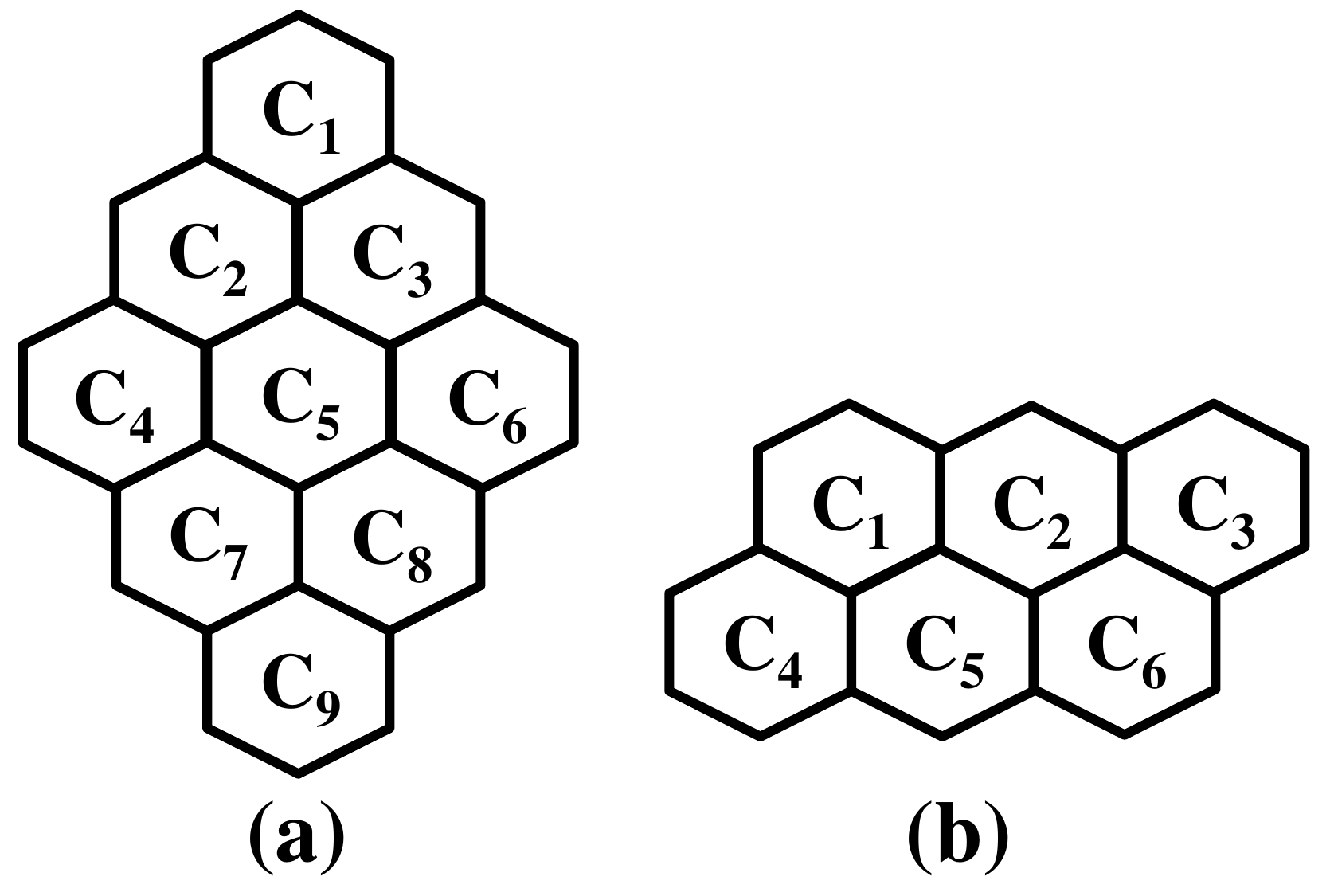}
\caption{Two different scenarios: (a) 9-cell CN (b) 6-cell CN.}
\label{three-scenarios}
\end{figure}

\begin{table}[h]
%\scriptsize
\centering
\caption{The UPS with 42 patterns for the 9-cell network. The patterns labeled with *  are constructed by Algorithm~\ref{pattern-generation-alg}.}
\label{patterns-42-22}
%\scalebox{1}{
\begin{tabular}{lll|lll} \toprule
$m$  & $i: I_i \in P_m$  & $i: O_i \in P_m$  &$m$ & $i: I_i \in P_m$ & $i: O_i \in P_m$ \\  \toprule
 1*  &  $1,\ldots,9$     & $--$              &22* & $6,7,8,9$        & $2$              \\
 2*  &  $1,\ldots,6$     & $9$               &23  & $6$              & $2,9$            \\
 3*  &  $1,2,3,4$        & $8$               &24* & $6$              & $2,7$            \\
 4*  &  $1,2,3,6$        & $7$               &25  & $--$             & $2,8$            \\
 5*  &  $1,2,4,7,9$      & $6$               &26* & $7,9$            & $2,6$            \\
 6   &  $1,2,4$          & $6,9$             &27  & $--$             & $2,6,9$          \\
 7*  &  $1,2$            & $6,7$             &28* & $--$             & $2,6,7$          \\
 8*  &  $1,3,6,8,9$      & $4$               &29* & $4,\ldots,9$     & $1$              \\
 9   &  $1,3,6$          & $4,9$             &30* & $4,5,6$          & $1,9$            \\
 10* & $1,3$             & $4,8$             &31  & $4$              & $1,8$            \\
 11* & $1,9$             & $5$               &32  & $4,7,9$          & $1,6$            \\
 12* & $1$               & $5,9$             &33  & $4$              & $1,6,9$          \\
 13  & $1,9$             & $4,6$             &34  & $6$              & $1,7$            \\
 14  & $1$               & $4,6,9$           &35  & $--$             & $1,6,7$          \\
 15* & $4,7,8,9$         & $3$               &36* & $9$              & $1,5$            \\
 16  & $4$               & $3,9$             &37* & $--$             & $1,5,9$          \\
 17* & $4$               & $3,8$             &38  & $6,8,9$          & $1,4$            \\
 18  & $--$              & $3,7$             &39  & $6$              & $1,4,9$          \\
 19* & $8,9$             & $3,4$             &40  & $--$             & $1,4,8$          \\
 20  & $--$              & $3,4,9$           &41  & $9$              & $1,4,6$          \\
 21* & $--$              & $3,4,8$           &42  & $--$             & $1,4,6,9$        \\  \bottomrule

\end{tabular}%}
\end{table}

\begin{table}[h]
%\scriptsize
\centering
\caption{The UPS with 13 patterns for the 6-cell network. The patterns labeled with *  are constructed by Algorithm~\ref{pattern-generation-alg}. }
\label{patterns-13-10}
%\scalebox{0.9}{
\begin{tabular}{lll|lll} \toprule
$m$  & $i: I_i \in P_m$  & $i: O_i \in P_m$  &$m$ & $i: I_i \in P_m$ & $i: O_i \in P_m$ \\  \midrule
 1*  &  $--$             & $1,6$             &8*  & $1,4,5$          & $3$              \\
 2*  &  $--$             & $2,4$             &9*  & $3,6$            & $1$              \\
 3*  &  $--$             & $3,5$             &10  & $--$             & $3,4$            \\
 4*  &  $1,\ldots,6$   & $--$              &11  & $--$             & $4,6$            \\
 5*  &  $2,3,6$          & $4$               &12* & $1,4$            & $6$              \\
 6*  &  $4$              & $2$               &13  & $--$             & $1,3$            \\
 7*  &  $3$              & $5$               &  & $$             & $$             \\  \bottomrule
\end{tabular}%}
\end{table}

\section{The Proposed Multi-cell Scheduler}
\label{section3}
In this section, we assume that the pattern set $\{ P_m \}_{m=1}^M$ is a-priori known whether be EPS-$n$, or the UPS if available, or the GPS produced by the algorithm presented in the previous section, or any other GPS. Given $\{ P_m \}_{m=1}^M$,  the scheduler we propose will go through the following  process at the beginning of each time slot
$\tau$:
\begin{enumerate} [label=(\roman*),leftmargin=10mm]
\item at the lower level, each BS will nominate one user for each of its inner and outer sections and their available transmission rates to the central controller,
\item at the upper level, the central controller computes the potential overall transmission rate for each pattern and decides to activate one of them, say $P_{m^*(\tau)}$, from the $M$ available patterns in the GPS, and mutes all other patterns,
\item at the lower level and for each section included in the chosen  pattern $P_{m^*(\tau)}$, the nominated user from step (i) is activated to receive traffic from its associated BS.
\end{enumerate}
Note that the decision is not made by the central controller in step (i), but is instead made by the local base station improving the scalability of the proposed architecture.
In this paper, such pattern activation and user selection decisions are made with throughput optimization in mind but also taking into consideration certain temporal fairness constraints. For such constraints, we need definitions for per-pattern and per-user temporal shares which are provided in the next subsection.
\subsection{Per-pattern and Per-user Temporal Shares}
Let $A_m$, $a^I_{i,j}$, and $a^O_{i,j}$ denote the long-term temporal share of pattern $P_m$, user $U^I_{i,j}$ and user $U^O_{i,j}$, respectively:
\begin{eqnarray*}
   A_m & = &\lim_{t\to\infty}   \frac{1}{t} \sum_{\tau=1}^t \mathbbm{1}{\{m^*(\tau)=m\}} \label{patternshare} \\
   a^I_{i,j} & = &\lim_{t\to\infty}  \frac{1}{t}\sum_{\tau=1}^{t}\mathbbm{1}{\{\text{user}~U^I_{i,j}~\text{is scheduled at slot}~\tau\}} \label{usershare-I} \\
   a^O_{i,j} & = & \lim_{t\to\infty} \frac{1}{t}\sum_{\tau=1}^{t}\mathbbm{1}{\{\text{user}~U^O_{i,j}~\text{is scheduled at slot}~\tau\}} \label{usershare-O}
\end{eqnarray*}
where $\mathbbm{1}{\{ \cdot \}}$ denotes the conventional indicator function which is either one or zero depending on whether the argument is true or not, respectively.
% We define the long-term air-time shares of pattern $P_m$, user $U^I_{i,j}$, and user $U^O_{i,j}$, respectively, as follows:
% \begin{eqnarray}
% a_m & = & \underset{t \rightarrow \infty}{\operatorname{lim}}~ a_m(t)  \label{nail1} \\
% a^I_{i,j} & = & \underset{t \rightarrow \infty}{\operatorname{lim}}~ a^I_{i,j}(t)   \label{nail2} \\
% a^O_{i,j} & = & \underset{t \rightarrow \infty}{\operatorname{lim}}~ a^O_{i,j}(t). \label{nail3}.
% \end{eqnarray}
Similarly, let $A^I_i$, $A^O_i$, and $A^C_i$ denote the long-term air-time share of section $I_i$, section $O_i$, and cell $C_i$, respectively. Mathematically,
\begin{eqnarray*}
A^I_i & = & \lim_{t\to\infty}  \frac{1}{t} \sum_{\tau=1}^{t} \mathbbm{1}{ \{ I_i \in P_{m^*(\tau)}\}}    \label{sectionshare-I}\\
A^O_i  & = &  \lim_{t\to\infty} \frac{1}{t} \sum_{\tau=1}^{t} \mathbbm{1}{ \{ O_i \in P_{m^*(\tau)}\}}   \label{sectionshare-O}\\
A^C_i & = & A^I_i +A^O_i \label{cellshare}.
\end{eqnarray*}
% We define the long-term air-time shares of cell $C_i$, section $I_i$, and section $O_i$ by $A_i = \underset{t \rightarrow \infty}{\operatorname{lim}}~ A_i(t)$, $A^I_i = \underset{t \rightarrow \infty}{\operatorname{lim}}~ A^I_i(t)$, and $A^O_i = \underset{t \rightarrow \infty}{\operatorname{lim}}~ A^O_i(t)$, respectively.
We introduce positive scheduling weights $\{w_m\}_{m=1}^M$ for each pattern $P_m$ satisfying $\sum_m w_m=1$. The scheduling weight $w_m$ is the long-term average probability that pattern $P_m$ is selected by the scheduler. Note that use of scheduling weight $\{w_m\}_{m=1}^M$ by the scheduler yields $A_m=w_m,\forall m$.
In this case, the system is said to be inter-pattern weighted temporal fair with respect to the weights $\{w_m\}_{m=1}^M$. In the specific case of $w_m=1/M$, the system is called inter-pattern temporal fair.
Similarly, we introduce positive scheduling weights $\{w^I_{i,j}\}_{j=1}^{N^I_i}$ and $\{w^O_{i,j}\}_{j=1}^{N^O_i}$ for user $U^I_{i,j}$ and user $U^O_{i,j}$, respectively, satisfying $\sum_j w^I_{i,j}=1$ and $\sum_j w^O_{i,j}=1$ for each of the cells $C_i$.
When $a^I_{i,j}=A^I_i w^I_{i,j}, \forall j$ for any section $I_i$ and $a^O_{i,j}=A^O_iw^O_{i,j}, \forall j$ for any section $O_i$, then we have intra-section weighted temporal fairness in sections $I_i$ and $O_i$, respectively, with respect to the weights $\{w^I_{i,j}\}$  and $\{w^O_{i,j}\}$.
% These three sets of weights, namely $\{w_m\}$, $\{w^I_{i,j}\}$, and $\{w^O_{i,j}\}$), can be adjusted by network operator for the scheduler to define several types of services for users and it is the task of the scheduler to choose a patterns and a number of users for each time slot in a way that inter-pattern and intra-section weighted fairnesses get established in the network.
In this paper, we only consider ordinary intra-section temporal fairness which leads to the following two identities:
\begin{eqnarray}
 w^I_{i,j} & = & \frac{1}{N^I_i}, ~1 \leq j \leq N^I_i, \label{intra-section-fairness-inner-user-weight} \\
 w^O_{i,j} & = & \frac{1}{N^O_i}, ~ 1 \leq j \leq N^O_i. \label{intra-section-fairness-outer-user-weight}
\end{eqnarray}
In case when $N^I_i$ or $N^O_i$ is zero, no per-user scheduling weights are assigned for that particular section. Moreover, no downlink transmissions would take place due to the lack of a user in that section even if the included pattern is selected for transmission.
Therefore, the only unknowns to the scheduler in the numerical examples of the current study are the per-pattern weights. Once the weights are decided, then the following identities immediately hold:
\begin{equation}
A^I_i =  \sum_{m=1}^M w_m e_{m,i }, \quad
a^I_{i,j} = \frac{A^I_i}{N^I_i}, \label{sectionshare-I2}
\end{equation}
and
\begin{equation}
A^O_i  =  \sum_{m=1}^M w_m f_{m,i}, \quad
a^O_{i,j}  = \frac{A^O_i}{N^O_i}, \label{sectionshare-O2}
\end{equation}
 where  $e_{m,i} \triangleq \mathbbm{1}{\{I_i \in P_m\}}$ and $f_{m,i} \triangleq \mathbbm{1}{\{O_i \in P_m\}}$.
In the next subsection, we focus on methods of obtaining these weights leading to
two different forms of temporal fairness  being sought in the CN.
\subsection{Temporal Fairness Criteria}
\label{subsec:ISPTF-MMTF}
In this paper, we consider two different TF criteria for the multi-cell CN for which the per-pattern weights can easily be obtained at reasonable computational complexity.
The first TF criterion is the so-called Inter-Section Proportional Temporal Fairness (IS-PTF) in which the CN is inter-cell temporal fair but an inner section of each cell receives a temporal share proportional with the temporal share of the outer section of the same cell using the same network-wide proportionality constant. Mathematically, there holds
 $A^C_i =A^C_j, 1 \leq i,j \leq K$ and $A^I_i  = d A^O_i, 1 \leq i \leq K$ for a fairness proportionality constant $d$ to be chosen by the network operator.
% In fact, the first equation in (\ref{FOS1}) is the inter-cell fairness constraint which says that the long-term air-time share of the cells should be equal. The second equation in (\ref{FOS1}) is the inter-section weighted fairness constraint and establish a sort of fairness within each cell between the inner and outer sections.
As a matter of fact, employing $d>1$ causes the inner sections to be scheduled more often than the outer sections and vice versa for the case $0<d<1$.
% The third equation is the sum weight constraint which should hold in any case. It is clear that the equations in (\ref{FOS1}) are in terms of the vector weight $\overline{W}$. Therefore, given the pattern set, solving these equations leads to a particular $\overline{W}$ with which FOS1 is satisfied.
It is clear that  IS-PTF can not be achieved for some general pattern sets such as GPS of Fig.~\ref{GPS-eg}. However, as stated in the theorem below, IS-PTF is achievable when EPS-$n$ is used as the pattern set.
\begin{Theorem}
With EPS-$n$ used as the pattern set,  IS-PTF is achieved by the following choice of weights:
\begin{align}
 &w_i  = \frac{1}{d+n} \quad 1 \leq i \leq n,  \label{inter-cell-fairness1}\\
 &w_{n+1}  =  \frac{d}{d+n},  \label{inter-cell-fairness2}
\end{align}
where $M=n+1$ and the pattern $P_{M}$ is the set of all inner sections of the CN, i.e.,  $P_{M}= \bigcup_{i=1}^K I_i$.
\end{Theorem}
\begin{proof}
We note that EPS-$n=$ MPS-$n \: \cup \: P_M$ and $A_m, 1\leq m \leq n$ and $A_M$ are the airtime shares of $m^{th}$ individual MPS-$n$ and $P_M$, respectively. Furthermore, the scheduler can guarantee that $A_m=w_m, 1\leq m \leq n$ and $A_M=w_M$.
Patterns in MPS-$n$ are mutually exclusive, thus for each section $O_i$ there is only one index $m, \:1\leq m \leq n$ such that  $A^O_i=w_m$. Furthermore, $\forall i: A^I_i=w_M$. We thus conclude that for all $i$, there exists an $m,~  1 \leq m \leq n$ such that $A^C_i=w_m+w_M$.
Recall that inter-cell fairness requires $\forall m_1,m_2,~ 1 \leq m_1,m_2 \leq n:~ w_{m_1}+w_M=w_{m_2}+w_M$ which in turn implies that $\forall m_1,m_2,~ 1 \leq m_1,m_2 \leq n: ~w_{m_1}=w_{m_2} \triangleq w$.
Moreover, using inter-section fairness we have $\forall m,~ 1 \leq m\leq n:w_M=d w_m=d w$. Finally, the pattern weights sum to unity, i.e. $\sum_{m=1}^nw_m+w_M=1$. Therefore, $nw+dw=1$ which implies that $\forall m,~ 1 \leq m\leq n: w_m=w=\frac{1}{d+n}$ and $w_M=w_{n+1}=\frac{d}{d+n}$.
\end{proof}
For the special case of EPS-3, there are $M=4$ patterns depicted in  Fig.~\ref{patternFFR3} where the corresponding weights are: $w_1=w_2=w_3=\frac{1}{d+3}$ and $w_4=\frac{d}{d+3}$.

The second TF criterion we study in this paper is Max-Min Temporal Fairness (MMTF) for which the per-pattern weights are selected so that
the minimum user temporal share in the CN is maximized.
% To this end first we should derive the users' temporal share in terms of $\overline{W}$. As it is mentioned earlier, we assume ordinary intra-section fairness ((\ref{intra-section-fairness-inner-user-weight}), (\ref{intra-section-fairness-outer-user-weight})) which leads to have the following equalities:
% \begin{eqnarray}
%  & a^I_{i,j}=\frac{A^I_i(\overline{W})}{N^I_i}, ~~j\in\{1,2,...,N^I_i\}, \label{inner-user-temporal-share} \\
%  & a^O_{i,j}=\frac{A^O_i(\overline{W})}{N^O_i}, ~~j\in\{1,2,...,N^O_i\}, \label{outer-user-temporal-share}
% \end{eqnarray}
% which are the temporal share of $U^I_{i,j}$ and $U^O_{i,j}$ in terms of $\overline{W}$, respectively. Using (\ref{inner-user-temporal-share}) and (\ref{outer-user-temporal-share}), we can define FOS2 as follows:
The MMTF is easily shown to be reducible to the following linear program:
\begin{align}
\text{max} \quad
z &  \label{MMTF}\\
\quad \text{\emph{subject to:}} \notag\\
\quad &z  \leq \frac{A^I_i}{N^I_i},\quad z  \leq \frac{A^O_i}{N^O_i} \quad &1 \leq i \leq  K, \notag\\
\quad &A^I_i  = \sum_{m=1}^M w_m e_{m,i} \quad &1 \leq i \leq  K, \notag\\
\quad &A^O_i  = \sum_{m=1}^M w_m f_{m,i} \quad &1 \leq i \leq  K, \notag\\
\quad &\sum_{m=1}^M w_m =1, \quad w_m \geq 0 \quad &1 \leq m \leq  M.\notag
\end{align}
% where  $e_{m,i} = \mathbbm{1}{\{I_i \in P_m\}}$ and $f_{m,i} = \mathbbm{1}{\{O_i \in P_m\}}$.

% \begin{eqnarray}
% FOS2:&  \nonumber\\
%  &\overline{W}^*= \underset{\overline{W}}{\operatorname{arg\:max \:min}}~~\{\frac{A^I_i(\overline{W})}{N^I_i}, \frac{A^I_i(\overline{W})}{N^O_i}\}_{i=1}^{K} \nonumber  \\
% &~s.t.~:~ \sum_m W_m=1,~~ \overline{W}>\overline{0},  \label{opt1}
% \end{eqnarray}
% which is not a tractable optimization problem in this form. By introducing an auxiliary variable,$Z$, we can convert (\ref{opt1}) into a linear programming problem as follows \cite{fourer_course}:

% \begin{eqnarray}
% &\overline{W}^*= \underset{\overline{W}}{\operatorname{arg\:max}}~~ Z \label{opt2} \\
% &~s.t.~:~ \nonumber \\
% &Z\leq \frac{A^I_i(\overline{W})}{N^I_i} : i\in\{1,2,\ldots, K\},  \nonumber \\
% &Z\leq \frac{A^O_i(\overline{W})}{N^O_i} : i\in\{1,2,\ldots, K\}, \nonumber  \\
% &\sum_m W_m=1,~~ \overline{W}>\overline{0}. \nonumber
% \end{eqnarray}
It is not difficult to show that for a given GPS, this linear program has at least one solution and in the case of a non-unique solution, the objective function value, i.e., $z$, is the same for all the solutions due to the convexity of the program. While there is no general closed form solution for MMTF, Theorem \ref{MMTF-solution} provides a closed form solution when the patterns are mutually exclusive (e.g. EPS-$n$). When patterns are not mutually exclusive, there are numerous efficient numerical methods including simplex and interior point algorithms via which one can obtain a solution for MMTF with reasonable effort.

\begin{Theorem}
\label{MMTF-solution}
  If the patterns are mutually exclusive, the solution of the MMTF problem is unique and we have
  \begin{align}
    w^*_m&=\frac{n_m}{\sum_{l=1}^{M}n_l} \quad 1 \leq m \leq M, \label{mutually-ex-w}\\
    z^*&=\frac{1}{\sum_{l=1}^{M}n_l}, \label{mutually-ex-z}
  \end{align}
  where $n_m$ denotes the number of users in the most crowded section of pattern $m$, i.e., $n_m=max_i\{N^I_ie_{m,i},N^I_if_{m,i}\}_{i=1}^K$.
\end{Theorem}

\begin{proof}
Let $\{w^*_m\}_{m=1}^M$ and $z^*$ denote the solution of MMTF problem (\ref{MMTF}). Because the patterns are mutually exclusive, each section of the network is covered by one pattern. Therefore, according to (\ref{sectionshare-I2}), the minimum user airtime in pattern $m$ is $w^*_m/n_m$. In other words, users in the most crowded section of pattern $m$ have the minimum airtime share among all the users in that pattern. The claim is that the value of $w^*_m/n_m$ is the same for every $m$. Let us assume that the claim is not correct. Therefore, if the user with minimum airtime share is in pattern $m'$, then there exists at least one $m''$ for which $w^*_{m'}/n_{m'} < w^*_{m''}/n_{m''}$. It is clear that we can increase the user minimum airtime share in the network by decreasing $w^*_{m''}$ and increasing $ w^*_{m'}$. Therefore, $w^*_{m''}$ and $w^*_{m'}$ are not optimal which is contradiction and the claim is correct. We can conclude that $\forall m:~z^*=w^*_m/n_m$. On the other hand, we know that $\sum_m w_m=1$. Eventually, we can find the unique solution of MMTF problem, i.e., $\{w^*_m\}_{m=1}^M$ and $z^*$, by solving the system of these $M+1$ independent linear equations as given in (\ref{mutually-ex-w}) and (\ref{mutually-ex-z}).
\end{proof}

We should note that unlike IS-PTF, the MMTF problem always has a solution for any GPS. We also note that to solve the MMTF problem,
each $BS_i$ sends the values $N_i^I$ and $N_i^O$ to the upper level which subsequently solves the linear program to obtain the per-pattern scheduling weights
$\{w_m\}_{m=1}^M$. Note that it is not necessary to solve the MMTF problem at each time slot. Instead, the MMTF problem can be solved when the number of users in any one of the sections of the CN changes. Alternatively, the MMTF linear program can be solved if the change in the number of users in the individual sections of the CN is substantial and the previous solutions's weights can be used until such a substantial change. It is clear that this alternative mechanism may reduce the computational burden on the central controller.
% So far, it was explained that using FOS1 or FOS2, the scheduler can find the patterns' scheduling weights ,$\{w_m\}_{m=1}^{M}$, so as to establish two well-known types of fairness in the network. As it is mentioned earlier, the scheduler task is to provide weighted inter-patter pattern fairness with respect to the weights $\{w_m\}_{m=1}^{M}$, and intra-section fairness. Note that if the pattern weights $\{w_m\}_{m=1}^{M}$ are obtained based on FOS1 ((\ref{inter-cellfairness}), (\ref{inter-section-weighted-fairness})), then weighted inter-pattern fairness is equivalent to the combination of inter-cell and inter-section fairnesses. On the other hand, if the pattern weights are obtained based on FOS2 (\ref{opt2}), then weighted inter-pattern fairness is equivalent to the Max-Min fairness.

For the two criteria for fairness studied in this paper, namely IS-PTF and MMTF, we have described methods by which the scheduling weights $\{w_m\}_{m=1}^M$ are obtained.
Exploration of other temporal fairness criteria are left for future research.
In the next subsection, we are going to introduce a two-level scheduler that uses these scheduling weights so as to make opportunistic decisions to select patterns at the upper level and also users belonging to these patterns at the lower level.
\subsection{Two-level Opportunistic Scheduler}
In this section, we assume that the scheduling weights $\{w_m\}_{m=1}^M$, $\{w^I_{i,j}\}_{j=1}^{N^I_i}$, and $\{w^O_{i,j}\}_{j=1}^{N^O_i}$ are a-priori known.
Although the per-user scheduling weights in the current study are based on the identities (\ref{intra-section-fairness-inner-user-weight}) and
(\ref{intra-section-fairness-outer-user-weight}) , the algorithm below works for arbitrary per-user scheduling weights $\{w^I_{i,j}\}_{j=1}^{N^I_i}$, and $\{w^O_{i,j}\}_{j=1}^{N^O_i}$.
Subsequently, we introduce a counter $b^I_{i,j}$ for each inner section user $U^I_{i,j}$, a counter $b^O_{i,j}$ for each outer section user $U^O_{i,j}$ and a counter $B_m$ for each pattern $P_m$. We set these counter values to zero at the beginning of network operation. We also define the instantaneous Spectral Efficiency (SE) $r^I_{i,j}(\tau)$ ($r^O_{i,j}(\tau)$) in units of bits/s/Hz for user $U^I_{i,j}$ ($U^O_{i,j}$) at time slot $\tau$. In particular, in our numerical experiments, we use the Shannon formula
\begin{eqnarray}
r^I_{i,j}(\tau) & = &\log_2(1+\text{SNR}^I_{i,j}(\tau)),\label{SE-I} \\
r^O_{i,j}(\tau) & =&\log_2(1+\text{SNR}^O_{i,j}(\tau)),\label{SE-O}
\end{eqnarray}
where $\text{SNR}^I_{i,j}(\tau)$ and $\text{SNR}^O_{i,j}(\tau)$ denote the signal to noise ratio of user $U^I_{i,j}$ and $U^O_{i,j}$, respectively at slot $\tau$ \cite{tse_book} but also other relationships of SE to the SNR than (\ref{SE-I}) and (\ref{SE-O}) can be used. Next, we describe the two-level multi-cell scheduling algorithm we propose at a given time slot $\tau$. At the cell level, for each cell $C_i$, its $BS_i$ selects two users $j^*_{i,I}$ and $j^*_{i,O}$ from the inner section $I_i$ and outer section $O_i$ of the cell, respectively, based on the instantaneous user SEs and user counter values as follows:
\begin{eqnarray}
j^*_{i,I} &  =&
{\operatorname{arg\, max}}_ {1 \leq j \leq  N^I_i }  \;\;  r^I_{i,j}(\tau)+\alpha  b^I_{i,j}, \label{user-selection-I}\\
j^*_{i,O} &  =&   {\operatorname{arg\, max}}_ {1 \leq j \leq  N^O_i }  \;\; r^O_{i,j}(\tau)+\alpha  b^O_{i,j}, \label{user-selection-O}
\end{eqnarray}
where $\alpha > 0$ is an algorithm parameter that we will study in the numerical examples to be shown to affect the convergence time and the overall network throughput.  For each cell, the $BS_i$ then nominates the users $j^*_{i,I}$ and $j^*_{i,O}$ and the instantaneous per-section SEs of the sections $I_i$ and $O_i$, denoted by
\begin{equation}
R^I_i(\tau) =r^I_{i,j^*_{i,I}}(\tau), \;\; R^O_i(\tau) =r^O_{i,j^*_{i,O}}(\tau), \label{instant}
\end{equation}
respectively, if the nominated users were to be served. In the second step, for all the cells, the $BS_i$ disseminates the values $R^I_i(\tau)$ and $R^O_i(\tau)$ to the upper level which is to subsequently select a pattern to be activated. For the pattern selection, in the third step, we obtain the network-wide SE of pattern $P_m$, denoted by $R_m(\tau )$  as follows:
\begin{equation}
R_m(\tau)=\sum_{i=1}^K \bigg( R^I_i(\tau) e_{m,i}+  R^O_i(\tau) f_{m,i}\bigg), ~ 1 \leq m \leq M.
\end{equation}
In the fourth step, the upper level selects the pattern $P_{m^*}$ on the basis of the following identity:
\begin{equation}
m^*  = \operatorname{arg\,max}_{1 \leq m \leq M } R_m(\tau)+\beta B_m, \label{pattern-selection}
\end{equation}
where $\beta$ is another algorithm parameter similar to $\alpha$ in (\ref{user-selection-I}) and (\ref{user-selection-O}). Once $m^*$ is determined, the pattern counters are updated in the fifth step as follows:
\begin{equation}
B_m =B_m+w_m-\mathbbm{1}{\{m=m^*\}}\quad 1\leq m\leq M. \label{pattern-bucket-update}
\end{equation}

The upper level then sends a message to all cells with the information on which pattern was selected in the current slot.
In the final step of the algorithm, the nominated users in the sections belonging to the selected pattern $P_{m^*}$ are scheduled in the current time slot.
Moreover, the counters of users in the sections belong to the selected pattern $P_{m^*}$ are updated as follows:
\begin{eqnarray}
 & b^I_{i,j}  =  b^I_{i,j}+w^I_{i,j}-\mathbbm{1}{\{j=j^*_{i,I}\}}, \; I_i \in P_{m^*}, \\
 & b^O_{i,j}  =  b^O_{i,j}+w^O_{i,j}-\mathbbm{1}{\{j=j^*_{i,O}\}}, \;O_i \in P_{m^*}.
 \end{eqnarray}
 %\label{user-bucket-update}
%\end{equation}
The counters of users in those sections that do not reside in the selected pattern are not updated and those sections are muted in the current slot.
The upper level scheduler has $\mathcal{O}(MK)$ computational complexity and $\mathcal{O}(M)$ storage requirements and presents a scalable solution when compared to existing methods whose complexity depend on the overall number of users $N$ in the network. Due to low communications overhead between the lower level BSs and upper level centralized entity, the proposed method is practical and can be implemented using a low-delay backhaul.
% The proposed scheme can be viewed as a two-level credit-based procedure where chosen (un-chosen) patterns loose (gain) credits and served (un-served) users of the chosen pattern again loose (gain) credits and the algorithm parameters $\alpha$ and $\beta$ are the weights of the credit component at the section and network levels, respectively.
%Moreover, the overall (at time $t$) network throughput $T(t)$ and the long-term average network throughput $T$ are defined as follows:
%\begin{align}
%T(t) &= \frac{BW}{t} \sum_{\tau=1}^t R_{m^*}(\tau),\\
%T  &= \lim_{t \rightarrow \infty} T(t).
%\end{align}

Moreover, the long-term average network throughput $T$ is defined as follows:
\begin{align}
T  = \lim_{t \rightarrow \infty} \frac{BW}{t} \sum_{\tau=1}^t R_{m^*}(\tau).
\end{align}
\subsection{Remarks on the Proposed Algorithm}
Per-pattern buckets are updated in
(\ref{pattern-bucket-update}) and it is obvious that these buckets can not grow to either plus or minus infinity. This can be observed from
the decision to activate a pattern made at the upper level according to (\ref{pattern-selection}) which ensures that buckets will stay bounded. Bounded buckets are indication of satisfaction of temporal fairness constraints at the pattern level. Similar conclusions can be drawn for the per-user buckets and the satisfaction of temporal fairness constraints at the user level. Moreover, the optimality of the proposed multi-cell scheduler stems from
the structure of the two TF schedulers, the network-level TF scheduler (\ref{pattern-selection}) and the cell-level TF scheduler (\ref{user-selection-I}) and (\ref{user-selection-O}), which are the same as the single-cell optimum TF scheduler described
in \cite{shroff_jsac01} except that we use fixed coefficients $\alpha$ and $\beta$ as in \cite{shahsavari_akar_wcl15} instead of those that decay in time. The purpose of this choice is to satisfy fairness constraints not only in the long term but also in shorter time scales. Numerical examples
will be presented to validate these choices.

In the proposed algorithm, as in \cite{tsePF, shroff_jsac01} and \cite{shahsavari_akar_wcl15}, all the users are assumed to be persistent, i.e., they always have data to receive.
However, in reality, users will occasionally receive data in busy periods followed by idle periods during which there would not be any traffic destined to these users.
A solution to this can be obtained by each BS to disseminate the number of backlogged users in each of its two sections at each time slot to the central controller which finds the scheduling weights either by IS-PTF or MMTF formulations again at each slot. However, this increases the computational burden for the MMTF formulation on the central controller because in this case, the linear program will change more rapidly. On the other hand, since the IS-PTF formulation has an explicit solution for EPS-$n$, this would not lead to any burden. Consideration of dynamic user traffic patterns for the MMTF formulation are left for future research.

\section{Numerical Examples}
\label{section4}
In all the numerical examples, we use one of the 6-cell and 9-cell CNs provided in Fig.~\ref{three-scenarios} and an additional 37-cell CN depicted in Fig.~\ref{37cell}. The radii parameter $R_I$ and $R_C$ are assumed to be $0.5$ km and $1$ km, respectively, for all the CNs. The system frequency is assumed to be $2$ GHz and $BW$ is set to $20$ MHz. Noise power spectral density and noise figure are assumed to be $-174~dBm/Hz$ and $9~dB$, respectively. Transmission power of each BS to the inner and outer users are $30~dBm$ and $40~dBm$, respectively. It is straightforward to show that the choice of transmit powers and radii parameters $R_C$ and $R_I$ satisfy (\ref{O-Oneighborhood})-(\ref{I-Ineighborhood}) for $n=3$ which is our focus in this section. The large-scale fading channel coefficients are modeled based on the ‘COST-231’ model as $-140.7-35.2\log_{10}(d^{U}_{BS})+\Psi$ in $dB$ scale, where $d^{U}_{BS}$ is the distance between the corresponding user and BS, and $\Psi \sim N(0,\sigma^2_{\Psi})$ represents the log-normal shadowing effect. We assume that $\sigma_{\Psi}=4~dB$.  Rayleigh fading model is adopted for small-scale channel coefficient variations \cite{rappaport_book}. In the first two examples, we concentrate on IS-PTF. The next three numerical examples study the MMTF formulation.
\begin{figure}[h]
\centering
\includegraphics[width=7cm]{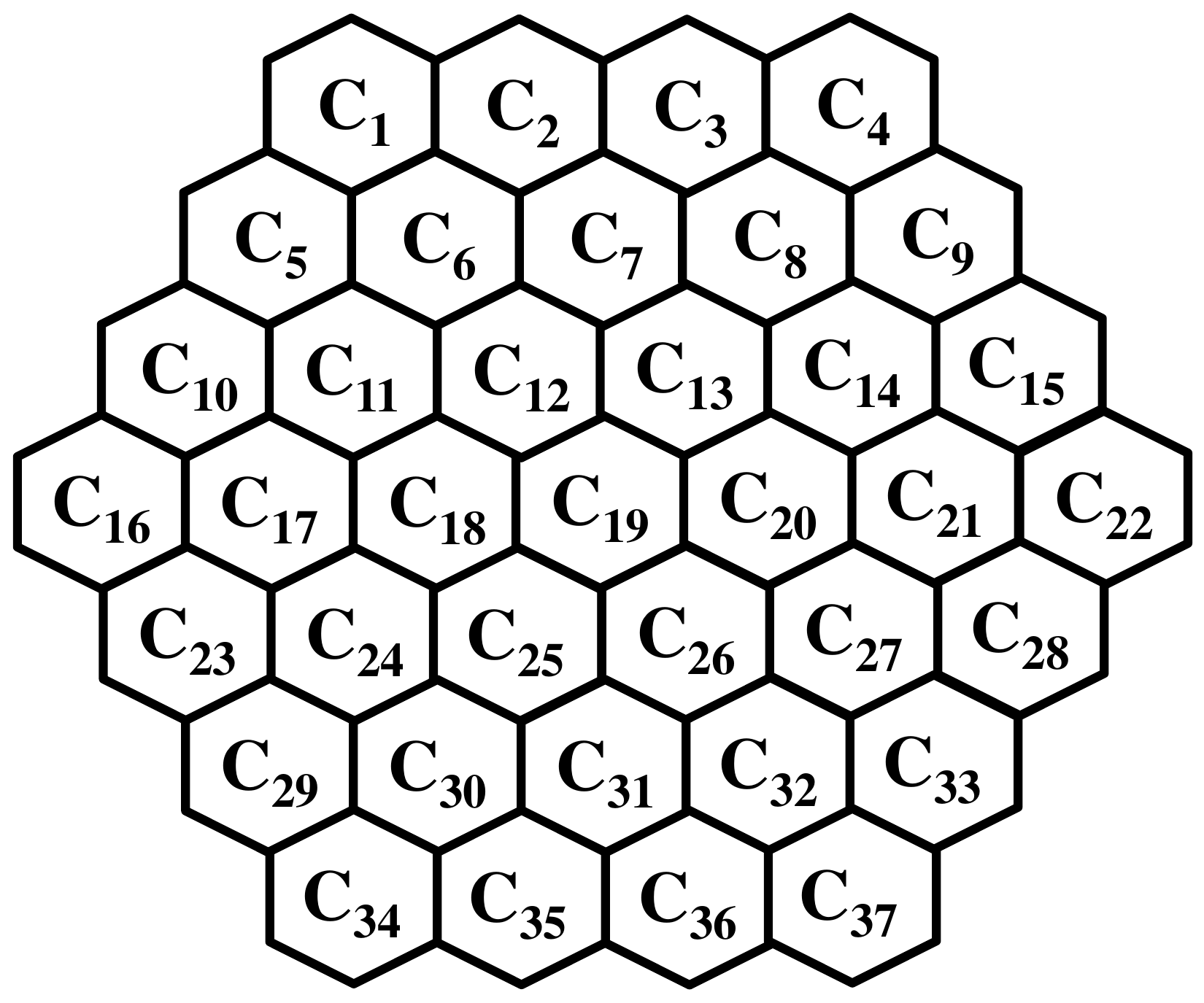}
\caption{37-cell cellular network}
\label{37cell}
\end{figure}
\subsection{Study of  Scheduler Parameters $\alpha$ and $\beta$}
In this example, we study the effect of the scheduler parameters $\alpha$ and $\beta$ employed through the identities
(\ref{user-selection-I}), (\ref{user-selection-O}), and (\ref{pattern-selection}), on the convergence time and long-term network throughput.
We assume that the scheduler uses EPS-3 with cardinality 4 in the 9-cell network of Fig.~\ref{three-scenarios}a.  We employ IS-PTF with the parameter $d$ set to unity,
 i.e., we seek ordinary inter-cell and inter-section fairness in this example. For a given section $I_i$ ($O_i$), let $J^I_i(t)$ ($J^O_i(t)$) denote the Jain's fairness index for the temporal shares  $a^I_{i,j}(t)$ ($a^O_{i,j}(t)$) which are shares of the users
$U^I_{i,j}$ $(U^O_{i,j}(t))$ up to time $t$.  We refer to \cite{jain_book} for the definition of Jain's fairness index.  Let us also define the
intra-section fairness $J^I(t)=min_iJ^I_i(t)$ ($J^O(t)=min_iJ^O_i(t)$) for inner (outer) sections.
Proximity of $J^I(t)$ ($J^O(t)$) to unity is representative of intra-section fairness for the inner (outer) section users up to time $t$.
Also, let the inter-pattern fairness index $J(t)$ be defined by Jain's fairness index for the individual per-pattern temporal shares $\{ A_m(t) \}$.
Similarly, proximity of $J(t)$ to unity is  representative of inter-pattern fairness up to time $t$. Furthermore, we note that inter-pattern fairness is equivalent to ordinary inter-cell fairness in this example. We assume $N=64$ uniformly located users in the 9-cell network of
Fig.~\ref{three-scenarios}a for each of the 20 simulation instances and for each instance we run the two-level scheduler for a duration of $5\times10^6$ slots with various choices of $\alpha$ and $\beta$. In each simulation instance, we obtain the values $\Gamma_J$ and $\Gamma_J^I$ ($\Gamma_J^O$), which are defined as the minimum value of $t$ such that $J(Ht)\geq 1-\epsilon$ and $J^I(Ht)\geq 1-\epsilon$ ($J^O(Ht)\geq 1-\epsilon$), respectively, for a small tolerance parameter $\epsilon>0$ which is set to $0.05$, and for a sampling parameter $H$ set to $1000$. A relatively large value of $\Gamma^I_J$ ($\Gamma^O_J$) is indicative of longer convergence times and therefore adverse impact on short-term intra-section fairness for inner (outer) section users. On the other hand, a relatively large value of $\Gamma_{J}$ is indicative of short-term inter-cell unfairness. The steady-state throughput $T$ and three fairness metrics $\Gamma_J$, $\Gamma_J^I$, and $\Gamma_J^O$ (average values obtained over the $20$ simulation instances) are tabulated in Table~\ref{throughput-fairness-table} for various choices of $\alpha$ and $\beta$.
According to Table~\ref{throughput-fairness-table}, while $\Gamma_J$ is more sensitive to the change of $\beta$, $\Gamma_J^I$ and $\Gamma_J^O$ are more sensitive to the change of $\alpha$. This is simply because coefficients $\beta$ and $\alpha$ directly appear in pattern and users selection levels, respectively. We also observe that with relatively smaller choices of the algorithm parameters $\alpha$ and $\beta$, the fairness indices converge slower and short-term inter-cell and intra-section fairness measures are compromised. The short-term inter-cell fairness measures appear to be less problematic. However, the total long-term network throughput $T$ is slightly improved with lower choices for  $\alpha$ and $\beta$. As a trade-off between short-term inter-cell and intra-section fairness and total network throughput, we fix $\alpha=\beta=0.01$ in the remaining numerical examples.

\begin{table}[h]
\centering
\caption{ The average throughput $T$, and the average of three fairness metrics $\Gamma_J$, $\Gamma_J^I$, and $\Gamma_J^O$ for various values of $\alpha$ and $\beta$ for the 9-cell CN.}
\label{throughput-fairness-table}
%\scalebox{0.9}{
\begin{tabular}{llllll} \toprule
$\alpha$      &  $\beta$     & $T(Mbps)$ &    $\Gamma_J$&      $\Gamma_J^I$&      $\Gamma_J^O$              \\ \midrule
             &  $0.1$        & $125.4$         & $1.00$         & $1.40$           &  $2.60$                \\ %\cmidrule(l){2-6}
$0.1$        &  $0.01$       & $125.5$         & $5.10$         & $1.05$           &  $3.00$                \\ %\cmidrule(l){2-6}
             &  $0.001$      & $125.8$         & $45.35$        & $1.00$           &  $8.00$                \\ \midrule

             &  $0.1$        & $126.4$         & $1.05$         & $11.45$          &  $20.05$                \\ %\cmidrule(l){2-6}
$0.01$       &  $0.01$       & $126.5$         & $4.90$         & $9.85$           &  $20.55$                \\ % \cmidrule(l){2-6}
             &  $0.001$      & $126.9$         & $45.35$        & $4.35$           &  $25.60$                \\ \midrule

             &  $0.1$        & $126.9$         & $1.15$         & $111.40$         &  $195.90$                \\ % \cmidrule(l){2-6}
$0.001$      &  $0.01$       & $127.1$         & $6.25$         & $110.00$         &  $196.50$                \\ % \cmidrule(l){2-6}
             &  $0.001$      & $127.4$         & $44.05$        & $94.35$          &  $201.00$                \\  \bottomrule
\end{tabular}%}
\end{table}

%\begin{table}[h]
%\centering
%\caption{ The average throughput $T$, and the average of three fairness metrics $\Gamma_J$, $\Gamma_J^I$, and $\Gamma_J^O$ for various values of $\alpha$ and $\beta$ for the 9-cell CN.}
%\label{throughput-fairness-table}
%\small
%\begin{tabular}{||c|c|c|c|c|c||} \hline \hline
%$\alpha$      &  $\beta$     & $T(Mbps)$ &    $\Gamma_J$&      $\Gamma_J^I$&      $\Gamma_J^O$              \\  \hline\hline
%             &  $0.1$        & $125.4$         & $1.00$         & $1.40$           &  $2.60$                \\  \cline{2-6}
%$0.1$        &  $0.01$       & $125.5$         & $5.10$         & $1.05$           &  $3.00$                \\  \cline{2-6}
%             &  $0.001$      & $125.8$         & $45.35$        & $1.00$           &  $8.00$                \\  \cline{1-6}
%
%             &  $0.1$        & $126.4$         & $1.05$         & $11.45$          &  $20.05$                \\  \cline{2-6}
%$0.01$       &  $0.01$       & $126.5$         & $4.90$         & $9.85$           &  $20.55$                \\  \cline{2-6}
%             &  $0.001$      & $126.9$         & $45.35$        & $4.35$           &  $25.60$                \\  \cline{1-6}
%
%             &  $0.1$        & $126.9$         & $1.15$         & $111.40$         &  $195.90$                \\  \cline{2-6}
%$0.001$      &  $0.01$       & $127.1$         & $6.25$         & $110.00$         &  $196.50$                \\  \cline{2-6}
%             &  $0.001$      & $127.4$         & $44.05$        & $94.35$          &  $201.00$                \\  \cline{1-6}
% \hline \hline
%\end{tabular}
%\end{table}
\subsection{Comparison of Opportunistic FFR vs. Benchmark FFR}
The use of the pattern set EPS-3 with the proposed opportunistic scheduler based on IS-PTF is referred to as Opportunistic FFR (OFFR) in this paper. In this example, OFFR uses the IS-PTF formulation with a certain choice of the parameter $d$ introduced in (\ref{inter-cell-fairness1}) and (\ref{inter-cell-fairness2}). In particular, we first study three different values of the parameter $d$, namely $d=1/4,1,4$ and consequently use the per-pattern weights as given in (\ref{inter-cell-fairness1}) and (\ref{inter-cell-fairness2}). The benchmark system called Benchmark FFR (BFFR) splits the $BW$ into four sub-bands, the bandwidth of each sub-band being proportional to
$w_m$ for $m\in\{1,2,3,4\}$ based on the identity (\ref{inter-cell-fairness1}) and (\ref{inter-cell-fairness2}).
The intra-cell scheduler of BFFR is the same as that of the OFFR. Besides, we note that there is no central controller in BFFR and ICI is handled by static spectrum partitioning as mentioned above.
For each value of $d$, we simulated BFFR and OFFR in both 9-cell and 37-cell CNs for a total of 400 instances each of which spans $10^6$ time slots.
In each of these simulations, the number of users $N$ is set to 64 and the users are randomly spread over the CN.
% On one hand if we can use the proposed scheduler with EPS-3 (which we call it Opportunistic FFR-3 (OFFR-3)) and on the other hand we can run the conventional FFR-3 system in the CN. We run these two methods in 9-cell and 37-cell CN and compare the results.
 Let $T^{OFFR}$ and $T^{BFFR}$ denote the overall network throughput when we employ OFFR and BFFR, respectively.
Also, let $T_{i,j}^{I,OFFR}$ and $T_{i,j}^{I,BFFR}$ denote the average throughput of user $U_{i,j}^I$ when we employ OFFR and  BFFR, respectively.
Similarly, let $T_{i,j}^{O,OFFR}$ and $T_{i,j}^{O,BFFR}$ denote the average throughput of user $U_{i,j}^O$ when OFFR and BFFR, respectively, are employed.
Furthermore, let
\[
G= \frac{T^{OFFR}-T^{BFFR}}{T^{BFFR}}100\%,
\]
denote the relative percentage network throughput gain of OFFR over BFFR.
Similarly, let
\begin{eqnarray*}
G_{i,j}^I & = &  \frac{T_{i,j}^{I,OFFR}-T_{i,j}^{I,BFFR}}{T_{i,j}^{I,BFFR}}100\%, \\
G_{i,j}^O & = &  \frac{T_{i,j}^{O,OFFR}-T_{i,j}^{O,BFFR}}{T_{i,j}^{O,BFFR}}100\%
\end{eqnarray*}
denote the percentage relative user throughput gain of OFFR over BFFR for users $U_{i,j}^I$ and $U_{i,j}^O$, respectively.
Figures~\ref{user-gain-cdf-9cell} and
\ref{user-gain-cdf-37cell} illustrate the empirical Cumulative Distribution Function (CDF) of the gains $G_{i,j}^I$ and $G_{i,j}^O$ for the specific case of $N=64$
for 9-cell and 37-cell CNs, respectively. We first note that as the parameter $d$ decreases, the patterns including outer users are scheduled more frequently in OFFR and a wider frequency band is assigned to outer section users in BFFR. In all studied cases, a large majority of users benefit from OFFR in comparison with BFFR. We also observe that when $d$ increases, outer section users gain substantially more with OFFR in comparison with inner users.  For example, when $d=4$, outer (inner) section users' throughput gains reach $47\%$ ($20\%$). On the other hand, when $d=1/4$, inner (outer) section users' throughput gains reach $57\%$ ($34\%$). The gains obtained in this example justify the benefits of the two-level scheduler.

% On the other hand, figs.~\ref{network-gain-cdf-9cell}) and \ref{network-gain-cdf-37cell}) depict the empirical CDF of the network gain $G$ for 9-cell and 37-cell CNs, respectively.

We further extend this example by varying $N \in \{32,64,128\}$ and the parameter $d \in \{1/4,1,4\}$ and simulate the same scenario.
Let $\mathbf{G}$, $\mathbf{G}^{\mathbf{I}}$ and $\mathbf{G}^{\mathbf{O}}$ denote the empirical means of the quantities $G$, $G^I_{i,j}$ and $G^O_{i,j}$, respectively, out of the 400 simulation instances. Also, let $\mathbf{P}$, $\mathbf{P}^{\mathbf{I}}$, and $\mathbf{P}^{\mathbf{O}}$ denote the percentage fraction of samples of $G$, $G^I_{i,j}$ and $G^O_{i,j}$, respectively, which are below zero, i.e., those scenarios or users who do not benefit from OFFR.  Table~\ref{DFFR-SFFR-table} provides the quantities  $\mathbf{G}$, $\mathbf{G}^{\mathbf{I}}$, $\mathbf{G}^{\mathbf{O}}$, $\mathbf{P}^{\mathbf{I}}$, and $\mathbf{P}^{\mathbf{O}}$, for varying choices of $d$ and $N$. We note that the quantity $\mathbf{P}$ were zero for all studied cases, i.e., all networks benefited from OFFR in comparison with BFFR in terms of average throughput.
Our findings are as follows:
\begin{itemize}[leftmargin=*]
\item{\textbf{Effect of $N$}}: We observe that the average gains $\mathbf{G}$, $\mathbf{G}^{\mathbf{I}}$, and $\mathbf{G}^{\mathbf{O}}$ appear to increase with decreasing average number of users per cell (for both CNs) in which case the multi-user diversity gain is limited with single-cell scheduling in BFFR. However, in such scenarios, network-wide multi-user diversity due to multi-cell scheduling helps improve user and network throughput in OFFR.
\item{\textbf{Effect of $d$}}: When $d$ increases, in both OFFR and BFFR systems, the whole frequency band is assigned to inner sections all the time. Therefore, for larger values of $d$, the data rate of the inner users is almost the same in both systems which explain the decrease in $\mathbf{G}^{\mathbf{I}}$ as $d$ increases. On the other hand, we observe that $\mathbf{G}^{\mathbf{O}}$ increases when $d$ increases. Furthermore, the behavior of $\mathbf{G}$ is similar to $\mathbf{G}^{\mathbf{I}}$ because data rate of inner users are typically higher than that of outer users.
\item{\textbf{Benefits to the users}}: A small fraction of users appear to lose with OFFR. On the other hand, up to $29.5\%$ and $23.5\%$ gains are observed in the empirical mean of inner and outer section users' throughput, respectively.
\item{\textbf{Benefits to the network operator}}: The quantity $\mathbf{P}$ was always zero for all the studied cases which means that OFFR consistently outperforms BFFR in terms of network throughput. In other words, OFFR is consistently beneficial from the network operator's perspective.
\end{itemize}

\begin{table}[h]
\centering
\caption{The comparative results associated with OFFR and BFFR when the system parameters $d$ and $N$ are varied.}
\label{DFFR-SFFR-table}
%\scalebox{0.9}{
\begin{tabular}{clllllll} \toprule
Scenario  &  $d$          &  $N$          & $\mathbf{G}$   & $\mathbf{G}^{\mathbf{I}}$ & $\mathbf{G}^{\mathbf{O}}$    & $\mathbf{P}^{\mathbf{I}}$  & $\mathbf{P}^{\mathbf{O}}$   \\  \midrule

          &               &  $32$         & $24.5$        & $32.6$         & $16.8$       & $0$     & $0.5$      \\
          & $\frac{1}{4}$ &  $64$         & $16.6$        & $21.1$         & $10.6$       & $0.2$   & $2.0$        \\
          &               &  $128$        & $11.1$        & $14.7$         & $6.5$        & $0.5$   & $6.8$      \\  \cmidrule(l){2-8}

          &               &  $32$         & $24.1$        & $20.6$         & $19.0$       & $0$     & $0.4$      \\
37-cell   &  $1$          &  $64$         & $16.2$        & $13.7$         & $11.9$       & $0.3$   & $2.1$      \\
          &               &  $128$        & $10.9$        & $8.8$          & $7.2$        & $1.1$   & $6.7$      \\  \cmidrule(l){2-8}

          &               &  $32$         & $19.8$        & $10.5$         & $23.5$       & $0$     & $0.4$      \\
          &  $4$          &  $64$         & $12.7$        & $7.0$          & $14.5$       & $0.1$   & $1.9$      \\
          &               &  $128$        & $8.0$         & $4.6$          & $8.9$        & $0.6$   & $6.1$      \\  \midrule

          &               &  $32$         & $22.0$        & $29.5$         & $14.5$       & $0.7$   & $4.3$    \\
          & $\frac{1}{4}$ &  $64$         & $14.8$        & $17.6$         & $9.0$        & $3.0$   & $10.0$   \\
          &               &  $128$        & $10.5$        & $11.5$         & $6.1$        & $6.1$   & $15.0$   \\  \cmidrule(l){2-8}

          &               &  $32$         & $22.1$        & $17.8$         & $16.1$       & $0.6$   & $4.3$    \\
9-cell    &  $1$          &  $64$         & $14.5$        & $10.5$         & $9.8$        & $4.4$   & $10.0$   \\
          &               &  $128$        & $9.9$         & $6.5$          & $6.7$        & $9.4$   & $14.5$   \\  \cmidrule(l){2-8}

          &               &  $32$         & $16.4$        & $9.4$          & $19.8$       & $0.5$   & $4.4$    \\
          &  $4$          &  $64$         & $10.3$        & $5.8$          & $12.3$       & $3.1$   & $8.9$    \\
          &               &  $128$        & $6.5$         & $3.4$          & $8.3$        & $9.1$   & $12.1$   \\  \bottomrule

\end{tabular}%}
\end{table}

 \begin{figure}[h]
 \centering

 \subfloat[9-cell scenario]{
         \label{user-gain-cdf-9cell}
         \includegraphics[width=0.6\linewidth]{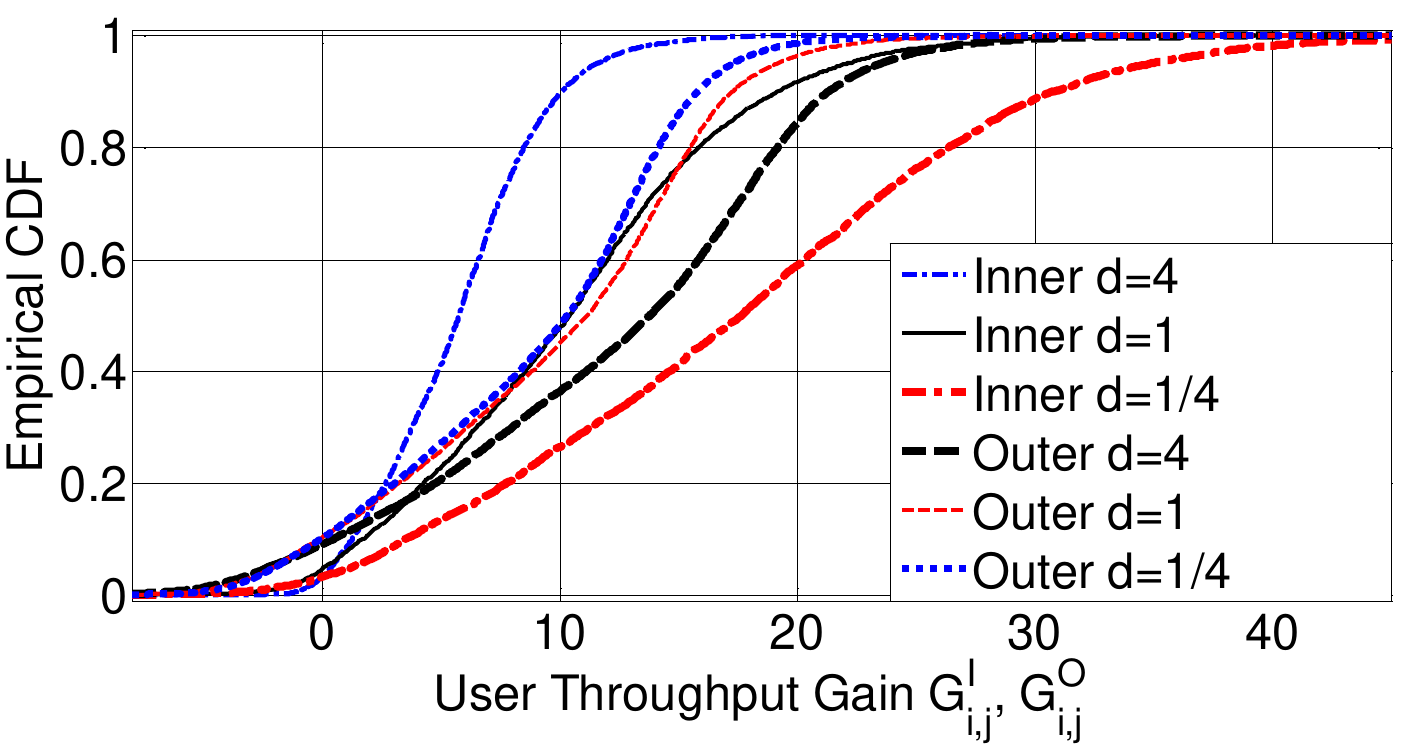}}

 \subfloat[37-cell scenario]{
         \label{user-gain-cdf-37cell}
         \includegraphics[width=0.6\linewidth]{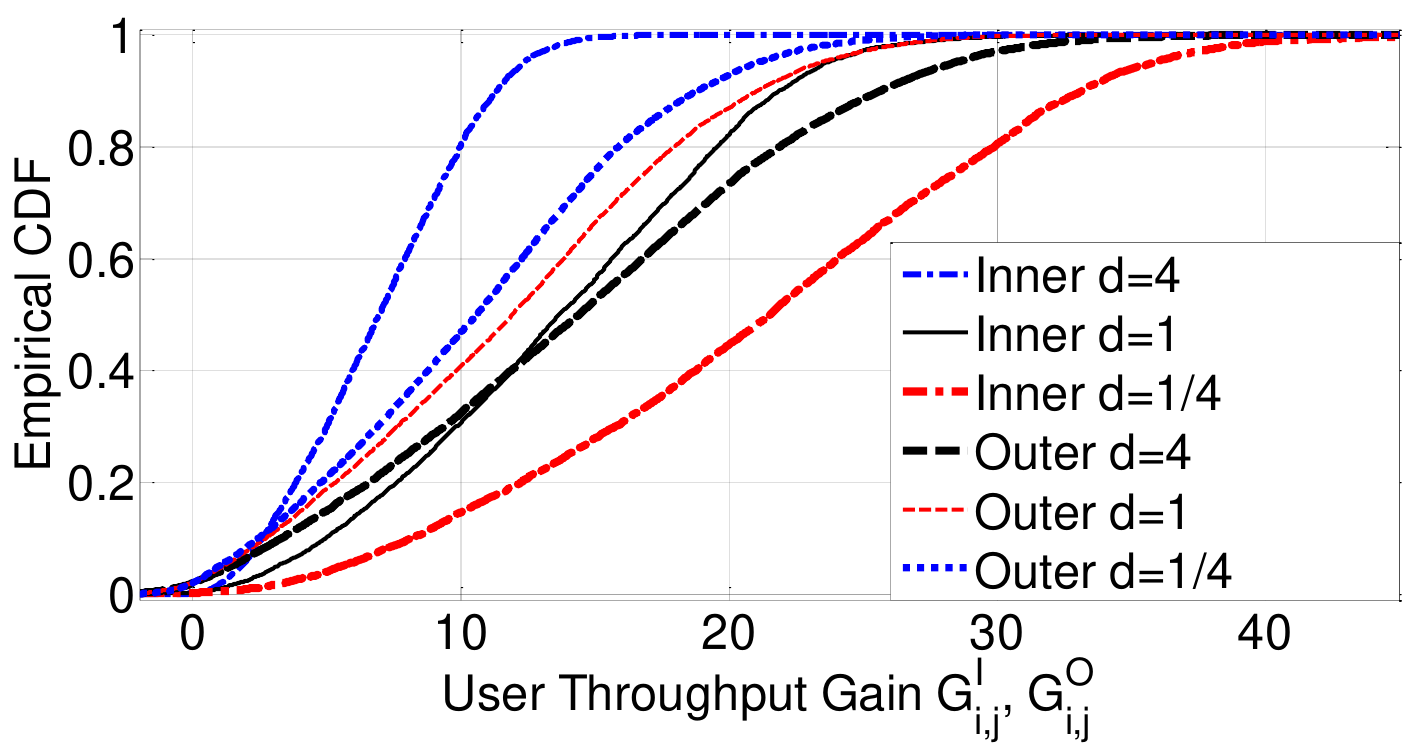} }

 \caption{Empirical CDF of the users throughput gains $G^I_{i,j}$ and $G^O_{i,j}$ for $N=64$ and different values of $d$ in 9-cell and 37-cell networks}
 \label{user-gain-cdf}

 \end{figure}

% \begin{figure}[h]
% \centering

% \subfloat[9-cell scenario]{
%         \label{network-gain-cdf-9cell}
%         \includegraphics[scale=0.4]{fig9.pdf}}

% \subfloat[37-cell scenario]{
%         \label{network-gain-cdf-37cell}
%         \includegraphics[scale=0.4]{fig10.pdf} }

% \caption{Empirical CDF of the networks throughput gain $G$ for different values of $d$ and $N$ in 9-cell and 37-cell networks}
% \label{kohler}

% \end{figure}

%\begin{figure}[h]
%\centering
%\includegraphics[scale=0.45]{fig7.pdf}
%\caption{Performance of different pattern sets in the 3-cell scenario as a function of $\eta$, the amount of non-uniformity in the network}
%\label{performance-3cell}
%\end{figure}

%\begin{figure}[h]
%\centering
%\includegraphics[scale=0.45]{fig8.pdf}
%\caption{Performance of different pattern sets in the 3-cell scenario as a function of $\rho$, the amount of non-uniformity in the network}
%\label{performance-3cell}
%\end{figure}
\subsection{Impact of GPS Selection for the MMTF Formulation}
\label{subsec:number-pattern}
In this section, we study the impact of the choice of the underlying pattern set in the context of the MMTF formulation (\ref{MMTF}). The performance metric is taken as the minimum of the temporal shares of all the users served in the CN ($z$ in (\ref{MMTF})). Recall that MMTF attempts to maximize this quantity through the linear program given in Section \ref{subsec:ISPTF-MMTF} through which  we obtain the per-pattern weights for this numerical example and consequently the performance metric. For this purpose, given the CNs with $K =6$ or $9$ cells depicted in Fig. \ref{three-scenarios}, $10K$ users are spread through the CN uniformly, leading to an average population of $10$ users per cell. After locating $10K$  users in the CN, one section is selected at random and $10P, P\in\{1,2,\ldots,10\}$ users are further introduced in this cell for the purpose of making the user distribution through the network more non-uniform.
To quantify this non-uniformity, we introduce  the parameter $\eta = (1+P)/(K+P)$ which gives the expected number of users in the most crowded cell divided by the overall number of users in the network. The larger the parameter $P$ or $\eta$, the more non-uniform the user distribution becomes.
Subsequently, each MMTF problem is solved 1000 times each of which is obtained by associating $10(K+P)$ users in the CN with the individual cells and their sections. The average of the 1000 instances is then reported.
We first study the 9-cell scenario given in Fig. ~\ref{three-scenarios}a.
We study the following pattern sets in the simulation study (the individual patterns are defined in Table \ref{patterns-42-22}):
\begin{itemize}%[leftmargin=*]
\setlength\itemsep{0.1mm}
\item $GPS_1=  \{ P_1,P_2,P_3,P_4,P_5,P_8,P_{11},P_{15},P_{22},P_{29}\}$
\item $GPS_2 = GPS_1 \bigcup \{P_{24},P_{36}\}$
\item $GPS_3 = GPS_2 \bigcup \{ P_6,P_{10},P_{17},P_{19},P_{21}\}$
\item $GPS_4$: pattern set obtained by Algorithm~\ref{pattern-generation-alg} with cardinality 22, introduced in Table \ref{patterns-42-22}.
\item EPS-3
\item UPS introduced by Table \ref{patterns-42-22}.
\end{itemize}
Fig.~\ref{performance-9cell} depicts the performance of the six pattern sets of choice mentioned above in terms of the minimum temporal share as a function of the non-uniformity parameter $\eta$. As $\eta$ increases, the minimum temporal share decreases for all the employed pattern sets. This observation stems from the fact that an increase in $\eta$ implies that one of the sections gets more crowded and consequently the users in that section get less chance to be scheduled.
UPS outperforms all the other pattern sets for all values of $\eta$ whereas the performance of $GPS_4$ is only slightly below that of UPS while the computational load required by Algorithm~\ref{pattern-generation-alg} to construct $GPS_4$ is remarkably lower than that of ES.  EPS-3 outperforms $GPS_3$ when $\eta$ is small but the situation is reversed when non-uniformity increases favoring the choice of $GPS_3$ over EPS-3. Furthermore, we observe that $GPS_3$ outperforms $GPS_2$ which in turn outperforms $GPS_1$. Therefore, we conclude that adding more patterns to an employed pattern set consistently enhances the performance of the MMTF scheduler. Similar conclusions are drawn for the 6-cell scenario for which we comparatively study four pattern sets, namely EPS-3, UPS, $GPS_1=\{P_1,P_2,P_3,P_4,P_5,P_6,P_7,P_8\}$, and $GPS_2 $ being the GPS obtained by Algorithm~\ref{pattern-generation-alg} where the patterns are illustrated in Table~\ref{patterns-13-10}. Fig.~\ref{performance-6cell} illustrates the performance of these four pattern sets. We conclude that pattern selection and/or its cardinality are crucial for MMTF schedulers especially for CNs with larger number of cells.
\begin{figure}[h]
\centering
\includegraphics[width=0.7\linewidth]{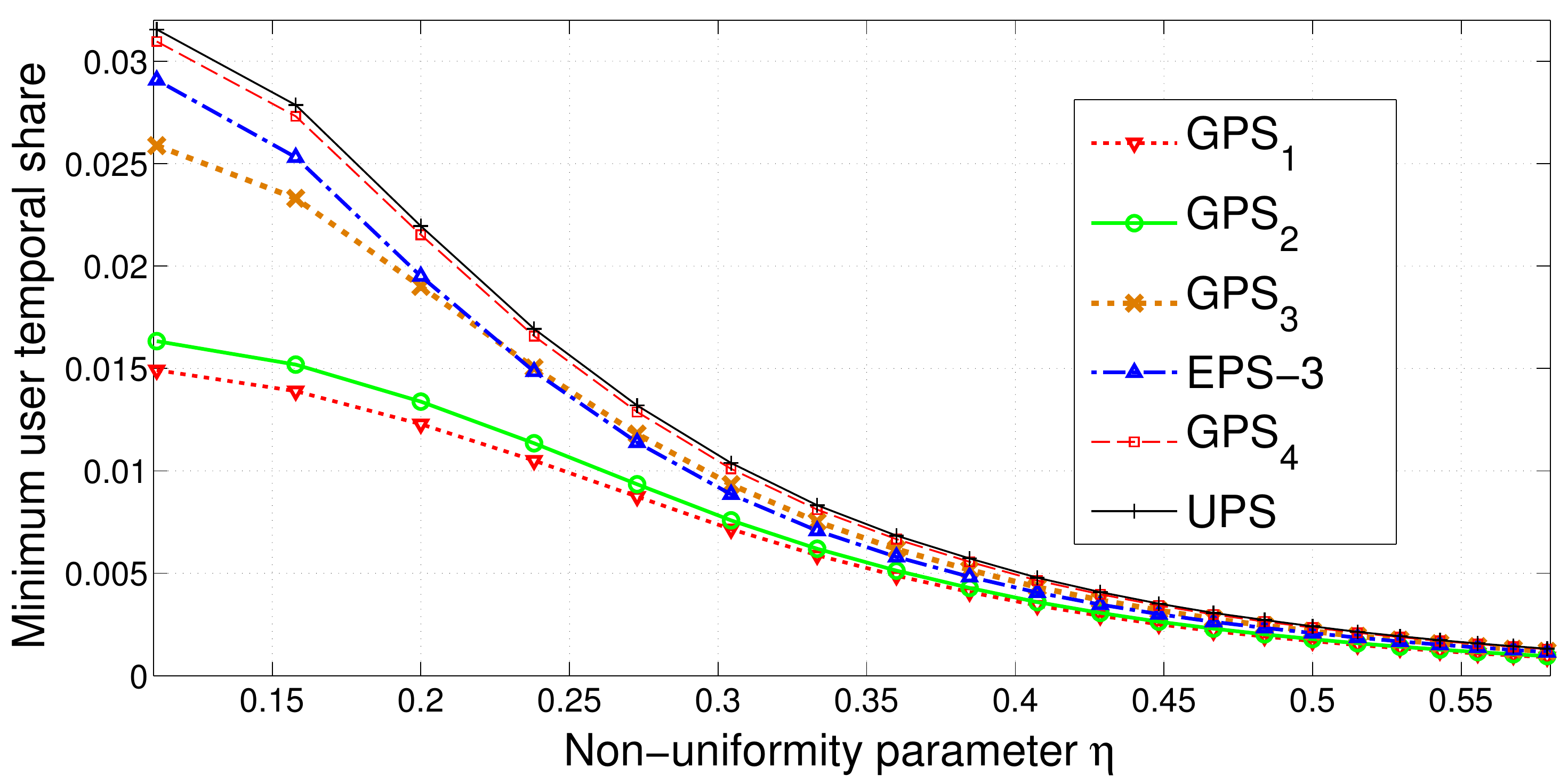}
\caption{Performance of various pattern sets in the 9-cell scenario as a function of the non-uniformity parameter $\eta$.}
\label{performance-9cell}
\end{figure}
\begin{figure}[h]
\centering
\includegraphics[width=0.7\linewidth]{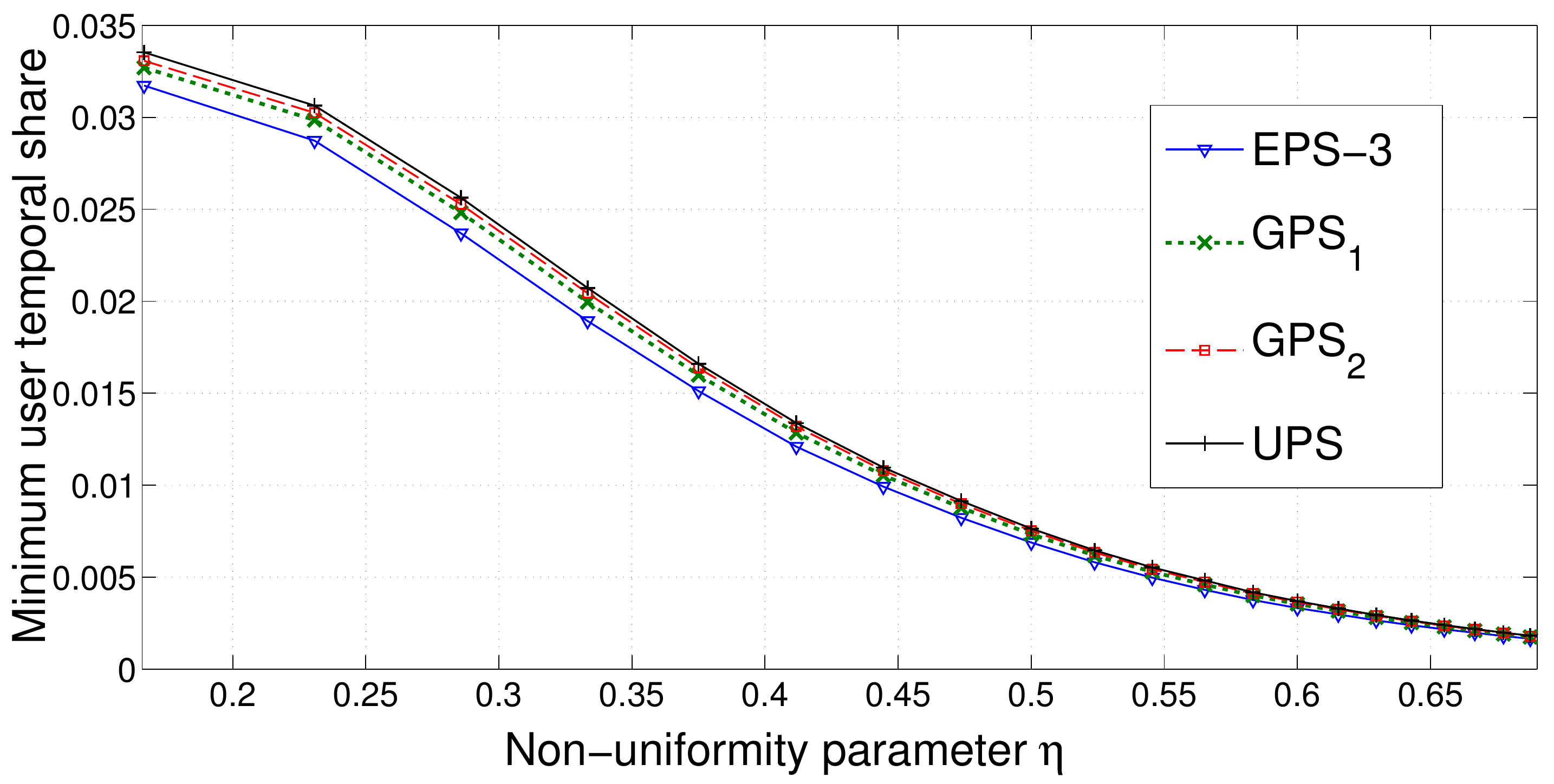}
\caption{Performance of various pattern sets in the 6-cell scenario as a function of the non-uniformity parameter $\eta$.}
\label{performance-6cell}
\end{figure}

% \begin{figure}[h]
% \centering
% \includegraphics[scale=0.6]{fig13.pdf}
% \caption{Performance of different pattern sets in the 3-cell scenario as a function of $\rho$, the amount of non-uniformity in the network}
% \label{performance-3cell}
% \end{figure}

% \subsection{Comparison of MMTF and IS-PTF}
% In this example, we compare the two forms of fairness provided in Section~II in terms of users' individual throughputs. To this end, we locate $N=64$ users at random
% in the 9-cell CN and run the opportunistic scheduler with the per-pattern weights found through the MMTF formulation for each instance.  For the pattern set, we employ the GPS constructed by Algorithm~\ref{pattern-generation-alg}. We also run the opportunistic scheduler with the EPS-3 pattern set and IS-PTF formulation
% for the choice of the proportionality parameter $d\in\{1,3,9\}$.  Each simulation is repeated 400 times.
% Figures~\ref{outer-user-rate} and \ref{inner-user-rate} depict the empirical CDF corresponding to outer and inner section  user throughputs, respectively.
% \begin{figure}[h]
% \centering
% \includegraphics[scale=0.55]{outer.pdf}
% \caption{Add Me}
% \label{outer-user-rate}
% \end{figure}
% \begin{figure}[h]
% \centering
% \includegraphics[scale=0.55]{inner.pdf}
% \caption{Add Me}
% \label{inner-user-rate}
% \end{figure}
\subsection{Performance of MMTF in Non-homogeneously Populated CNs}
In this example, we use Zipf-distributed user populations in the CN. In order to describe the Zipf distribution based on \cite{breslau1999web}, let $K$ be the number of elements,  and $k\in \{ 1,2,\ldots,K \}$ be the rank of an element. Let $s$ denote the parameter characterizing the Zipf distribution.  A given entity is an element of rank $k$ with probability $\frac{c}{k^s}, s \geq 0$ where $c$ is a normalization constant. When $K \rightarrow \infty$, the parameter $s$ should satisfy $s >1$.
When $K$ is finite and $s=0$, we have a discrete uniform distribution. When the parameter $s$ is increased, the population of the individual elements in the entire population becomes more non-uniform.
The Zipf distribution is successfully used in modeling various behaviors in computing and communications \cite{breslau1999web,shi_etal_ton05}.
In this example, we will study the impact of the parameter $s$ if a selected user $n \in \{1,2,\ldots,N \}$ resides in a cell $k\in \{1,2,\ldots,K \}$ according to a Zipf distribution with parameters $K$ and $s$.  We fix $K=9$ and we distribute $N \in \{32,64,128\}$ users in the 9-cell CN according to a Zipf distribution with parameter $s \in\{ 0,0.5,1,1.5,2\}$. The ranks of cells in the 9-cell network are depicted in Fig.~\ref{three-scenarios}a.  The user location within a cell is uniformly random. For each value of $N$, we distribute the users in the CN as described above and subsequently run the two-level scheduler using the MMTF formulation with $GPS_4$ in section \ref{subsec:number-pattern} with 22 patterns. We also run the scheduler using IS-PTF with EPS-3 and $d=1$ as a benchmark. For each pair of $s$ and $N$, we run the experiment 100 times each with a duration of $10^6$ time slots. We obtain the minimum user temporal share and minimum and average user throughput for each instance using MMTF and IS-PTF schedulers and take the averages of these individual values over all the 100 simulated instances for a given pair $(N,s)$.
Table~\ref{MMTF-ISPTF-table} illustrates the minimum user temporal share (denoted by $a_{min}$), the minimum user throughput (denoted by $T_{min}$), and the average user throughput (denoted by $T_{avg}$) for both MMTF and IS-PTF schedulers as a function of $s$ and $N$. In general, we observe that the MMTF scheduler always leads to higher minimum user throughput and temporal share. Also, when $N$ increases, the minimum user throughput and temporal share decrease for both schedulers as expected. Furthermore, when $s$ increases, the performance of the MMTF scheduler gets better relative to the IS-PTF scheduler in terms of $T_{min}$ and $a_{min}$. Therefore, we conclude that the MMTF scheduler becomes more effective in non-homogenously populated cellular networks if the performance metric is to be taken as the minimum user throughput or temporal share. We also observe that while the average user throughput is larger with IS-PTF in most of the cases, the gap between IS-PTF and MMTF is small.

\begin{table}[h]
\centering
\caption{Minimum user temporal share, minimum network-wide user throughput (in kbps) and average network-wide user throughput (in Mbps) using MMTF and IS-PTF schedulers in the 9-cell CN.}
\label{MMTF-ISPTF-table}
%\scalebox{0.9}{
\begin{tabular}{llllllll} \toprule
$s$          &  $N$         & $a^{MM}_{min}$ & $a^{IS}_{min}$ & $T^{MM}_{min}$  &  $T^{IS}_{min}$   & $T^{MM}_{avg}$  &  $T^{IS}_{avg}$  \\  \midrule

             &  $32$        & $0.093$      & $0.048$      & $207.1$      & $174.5$  & $3.10$  & $3.02$       \\
$0$          &  $64$        & $0.041$      & $0.028$      & $81.9$       & $71.9$   & $1.92$  & $2.08$       \\
             &  $128$       & $0.022$      & $0.015$      & $35.9$       & $35.5$   & $1.09$  & $1.24$    \\  \midrule

             &  $32$        & $0.066$      & $0.042$      & $184.4$      & $161.6$  & $2.96$  & $2.91$     \\
$0.5$        &  $64$        & $0.035$      & $0.023$      & $81.3$       & $77.1$   & $1.91$  & $1.96$    \\
             &  $128$       & $0.018$      & $0.012$      & $30.2$       & $30.1$   & $1.17$  & $1.22$    \\  \midrule

             &  $32$        & $0.052$      & $0.029$      & $180.5$      & $149.1$  & $2.81$  & $2.73$     \\
$1$          &  $64$        & $0.027$      & $0.015$      & $69.3$       & $55.3$   & $1.77$  & $1.78$    \\
             &  $128$       & $0.013$      & $0.007$      & $27.4$       & $24.1$   & $1.14$  & $1.15$     \\  \midrule

             &  $32$        & $0.043$      & $0.019$      & $131.5$      & $85.9$   & $2.27$  & $2.24$     \\
$1.5$        &  $64$        & $0.022$      & $0.010$      & $54.5$       & $39.2$   & $1.58$  & $1.60$   \\
             &  $128$       & $0.011$      & $0.005$      & $21.2$       & $16.3$   & $0.95$  & $0.97$    \\  \midrule

             &  $32$        & $0.038$      & $0.016$      & $130.4$      & $78.7$   & $2.02$  & $1.95$   \\
$2$          &  $64$        & $0.019$      & $0.007$      & $51.3$       & $31.9$   & $1.27$  & $1.30$   \\
             &  $128$       & $0.009$      & $0.004$      & $19.9$       & $11.7$   & $0.81$  & $0.84$    \\  \bottomrule

\end{tabular}%}
\end{table}

\subsection{Transient Behavior of the Proposed Scheduler}
In this example, we study the transient behavior of the proposed scheduler. At time $\tau=0$, we distribute $3$ and $9$ users in each inner and outer section of the 9-cell CN, respectively, at uniformly random. Subsequently, we run the proposed scheduler with the per-pattern weights obtained through the MMTF  formulation while employing the GPS constructed by Algorithm~1. At time slot $\tau_1=4\times10^5$, we add $3$ and $9$ more users to $I_5$ and $O_5$, respectively, and update the per-pattern weights using the MMTF formulation. At time slot $\tau_2=8\times10^5$, we remove these users back from the network and update back the per-pattern weights. Fig.~\ref{share-outer-user-cell5} illustrates the sampled cumulative temporal share of a randomly selected user in $O_5$ with sampling rate of $H=0.1$ (one sample per ten slots) as a function of  time whereas
Fig.~\ref{rate-outer-user-cell5} depicts the sampled average throughput of the same user as a function of time. We observe that the proposed scheduler converges quite rapidly with the employed scheduler parameters and the temporal share of this particular user approaches the max-min share as time evolves for any of the three time intervals.
\begin{figure}[h]
\centering

% \subfloat[Cumulative temporal share of an outer user in cell 2 in time]{
%         \label{share-outer-user-cell2}
%         \includegraphics[scale=0.4]{fig14.pdf}}

% \subfloat[Throughput of an outer user in cell 2 wrt time]{
%         \label{rate-outer-user-cell2}
%         \includegraphics[scale=0.4]{fig15.pdf} }

\subfloat[Cumulative temporal share of an outer user in cell 5 wrt time.]{
        \label{share-outer-user-cell5}
         \includegraphics[width=0.7\linewidth]{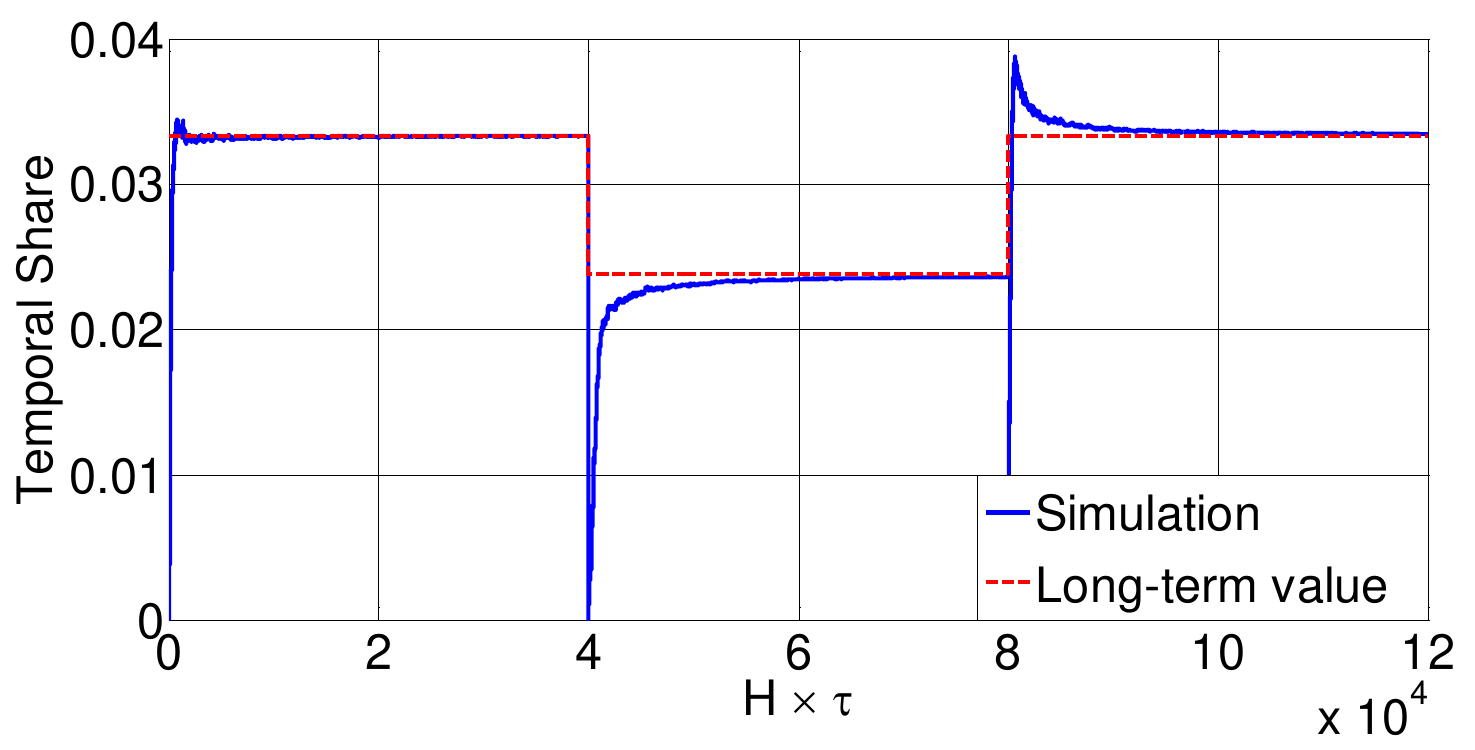}}

\subfloat[Throughput of an outer user in cell 5 wrt time.]{
        \label{rate-outer-user-cell5}
        \includegraphics[width=0.7\linewidth]{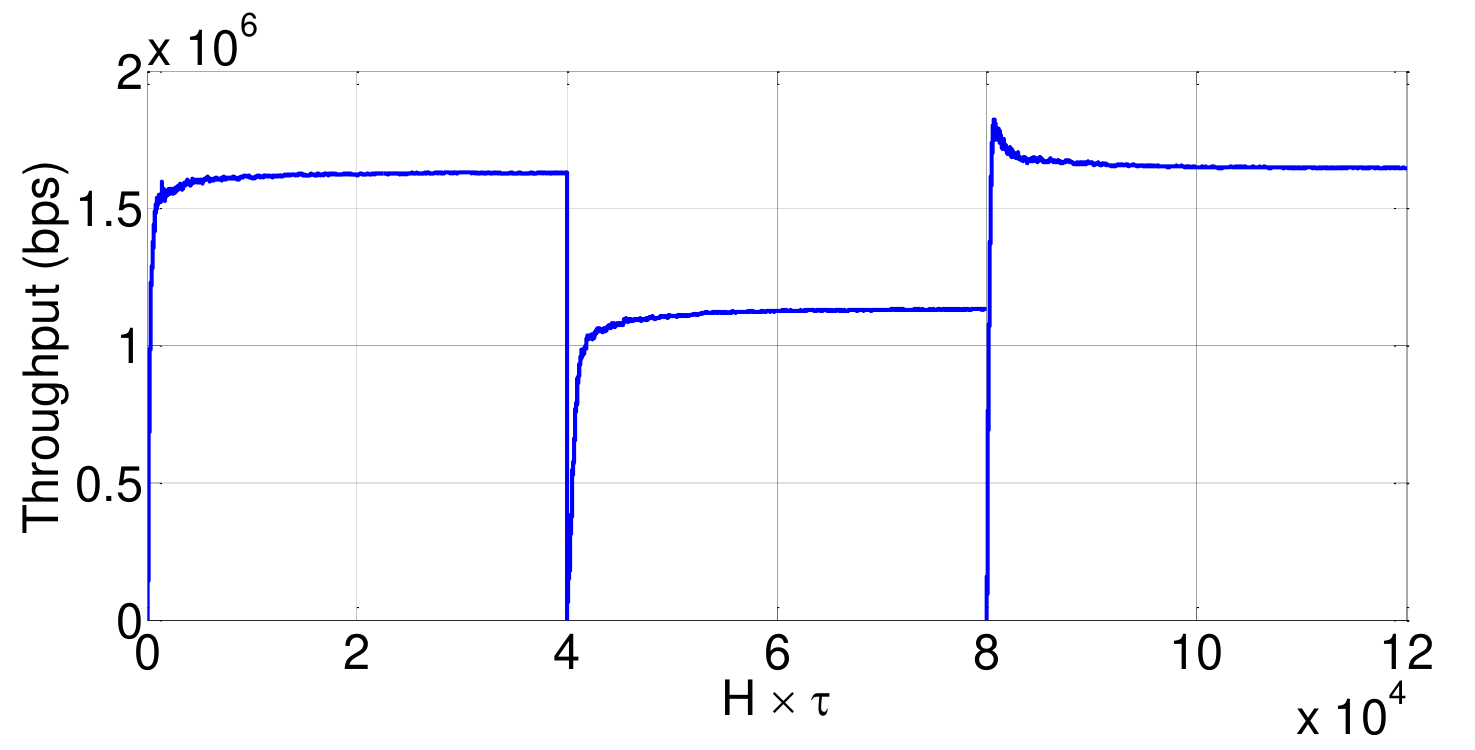} }

\caption{Transient behavior associated with a randomly chosen outer section user in Cell 5 in terms of a) temporal share and b) throughput as a function of time.
}
\label{user-gain-cdf}

\end{figure}

\section{Conclusions}
We have proposed a semi-centralized joint cell muting and user scheduling scheme for interference coordination in a multi-cell network under two temporal fairness criteria, namely IS-PTF and MMTF.  We have also proposed a novel cell muting pattern set construction algorithm  required for this joint scheme. For the IS-PTF formulation, we have shown that the proposed scheme, OFFR, outperforms the benchmark FFR system in in terms of network and users' throughput. On the other hand, the MMTF formulation allows one to perform dynamic load balancing with reasonable computational complexity. It has also been shown that the general pattern set to use and its cardinality also play a major role in the performance of the network-wide worst case user temporal share which is the performance metric we have used in this paper. For various cellular topologies and scenarios, we have shown that the pattern set we use by the proposed construction algorithm performs almost as well as using all possible patterns. Future work will consist of extending the methodology to more realistic OFDMA-based LTE networks and heterogeneous networks as well as incorporation of alternative fairness criteria and dynamic traffic models.

\ifCLASSOPTIONcaptionsoff
  \newpage
\fi

\bibliographystyle{IEEEtran}
\bibliography{fairness}

%\vfill


\end{document}